\documentclass[12pt]{article}
\usepackage[utf8]{inputenc}
\usepackage{epsfig}
\usepackage{amsmath,amsthm,amsfonts}
\usepackage{soul,color}
\usepackage{multirow}
\usepackage{bm, bbm}
\usepackage{epsfig}
\usepackage{natbib}
\usepackage{latexsym, graphicx, alltt}
\usepackage{graphicx}
\usepackage{setspace}
\usepackage{rotating}
\usepackage{threeparttable}
\usepackage{lscape}
\usepackage{geometry}
\usepackage[colorlinks=true,linkcolor=blue,urlcolor=blue,citecolor=blue,hyperfootnotes=false]{hyperref}
\usepackage[table]{xcolor}
\renewcommand{\baselinestretch}{1.2}

\newcommand{\Indic}{\mathbf{I}}

\newcommand{\pa}{\partial}         

\newcommand{\al}{\alpha}

\newcommand{\ga}{\gamma}

\renewcommand{\th}{\theta}

\newcommand{\convp}{\stackrel{p}\rightarrow}
\newcommand{\convd}{\stackrel{d}\rightarrow}

\newtheorem{cor}{Corollary}

\newtheorem{theorem}{Theorem}

\newtheorem{lemma}{Lemma}
\newtheorem{ass}{Assumption}

\newtheorem{remark}{Remark}

\newcommand\argmax[1]{\underset{#1}{\operatorname{argmax}}\;}

\def\baselinestretch{1.2}
\usepackage{bm}
\usepackage{epsfig}

\begin{document}

\title{Specification testing with grouped fixed effects \footnote{We
    are grateful to Francesco Bartolucci, Jeffrey Campbell, Giuseppe Cavaliere, Pavel \v{C}i\v{z}ek, Andreas Dzemski, Arturas Juodis, Riccardo Lucchetti, Elena Manresa, Chris Muris, Silvia Sarpietro Laura Serlenga, Amrei
    Stammann, Alexandros Theoludis, to the audience at the 28th
    International Panel Data Conference, at the 11th Italian Congress of Econometrics and Empirical Economics, and at the SEG seminar
    their helpful comments and suggestions. We also thank Carolina
    Castagnetti and Federico Belotti for generously sharing their
    codes.}}
\author{Claudia Pigini\footnote{Marche Polytechnic University
    (Italy).  Corresponding Author. Address: Department of Economics and Social Sciences, P.le Martelli 8, 60121, Ancona (Italy).E-mail: \url{c.pigini@univpm.it}} \and Alessandro
  Pionati\footnote{Marche Polytechnic University (Italy). E-mail:
    \url{a.pionati@univpm.it}} \and Francesco Valentini
  \footnote{University of Pisa (Italy). E-mail: \url{francesco.valentini@unipi.it}}}

\date{}
\maketitle

\def\baselinestretch{1.8}

\begin{abstract}
  We propose a Hausman test for the correct specification of
  unobserved heterogeneity in both linear and nonlinear fixed-effects
  panel data models. The null hypothesis is that heterogeneity is
  either time-invariant or, symmetrically, described by homogeneous
  time effects. We contrast the standard one-way fixed-effects
  estimator with the recently developed two-way grouped fixed-effects
  estimator, that is consistent in the presence of time-varying
  heterogeneity (or heterogeneous time effects) under minimal
  specification and distributional assumptions for the unobserved
  effects. The Hausman test compares jackknife corrected estimators,
  removing the leading term of the incidental parameters and
  approximation biases, and exploits bootstrap to obtain the variance
  of the vector of contrasts. We provide Monte Carlo evidence on the size and
power properties of the test and illustrate its application in
  two empirical settings.
 \end{abstract}

 \doublespacing 

\section{Introduction}

Correct specification of unobserved heterogeneity is crucial in panel
data modeling. For long, empirical applications have only considered
individual time-constant fixed effects, but the assumption of
time-invariant unobserved heterogeneity is often hardly tenable in
practice, especially over a long time dimension. Therefore the current
mainstream approach includes both subject- and time-specific intercepts, in
order to achieve credible identification of the effects of interest.
The simplest and most widely employed setup is the specification of
additive individual and time heterogeneity, namely the two-way
fixed-effects model, that in the linear model is equivalent to the
two-way correlated random effects approach \citep{wooldridge2021}. For
nonlinear models with additive fixed effects, \cite{FVW2016} provide
analytical and jackknife bias corrections for the maximum likelihood
(ML) estimator, which is plagued by the
incidental parameters problem.

While of simple implementation, the two-way fixed-effects
specification fails to capture the specific impact common
factors may have on each subject. There is now an important stream of
literature focused on developing identification results and estimation
strategies for models with interactive time and individual fixed
effects. Contributions have been spurred by the seminal paper by
\cite{Bai2009}, who provided identification results along with the
asymptotics for the interactive fixed-effects estimator in linear
models. More recently, interactive fixed-effects have been introduced
in nonlinear panel data and network models by \cite{chenfvweidner2021}.

Testing the assumptions on the unobserved heterogeneity specification
has also received considerable attention in the recent econometric
literature. \cite{BBP2015} propose a Hausman-type test for the
null hypothesis of time-constant unobserved heterogeneity in generalized
linear models, where conditional ML estimators are compared with
first-differences or pairwise conditional ML estimators. In the
context of large stationary panel models, the factor specification
could be tested by comparing additive to interactive fixed-effects
models, on the basis of the Hausman test suggested by \cite{Bai2009}
and its fixed-$T$ version, derived by \cite{Westerlund2019}. However,
it has been shown that the Hausman-type test fails to reject the null
hypothesis when individual factor loadings are independent across
equations \citep{CRT2015}. On this basis, \cite{kapetanios2023testing}
use a Hausman-type test contrasting additive and interactive
fixed-effects to detect such correlation, whereas
\cite{castagnetti2015} overcome the issue by proposing an alternative
max-type test for the null hypothesis of time-invariant unobserved
heterogeneity.

Despite its increasing popularity, the interactive effects approach
based on \citeauthor{Bai2009}'s procedure comes with some non trivial
issues. First, estimation relies on solving a non convex objective
function with possibly multiple minima
\citep{moon2023nuclear}. Secondly, the reliability of the iterative
procedure crucially depends on the consistency of the parameter
estimates chosen as the starting point for the algorithm
\citep{hsiao2018panel}. These drawbacks might be even more hampering
when nonlinear models with interacted fixed effects are involved
\citep[see][]{chenfvweidner2021}. In light of these considerations, a
simpler specification might therefore be preferable, provided it gives
a good enough representation of the structure of the unobserved
heterogeneity.

In this paper we propose a Hausman test for the fixed-effects
specification, in both linear and nonlinear models and where the
unobserved heterogeneity, under the null hypothesis, is only
individual or time-specific. The test contrasts One-Way Fixed Effects
(OW-FE henceforth) ML estimators with the Two-Way Grouped Fixed
Effects (TW-GFE henceforth) approach, recently put forward by
\cite{blm2022}. The TW-GFE approach is based on a first-step
data-driven approximation of the unobserved heterogeneity, which is
clustered by the {\em kmeans} algorithm using individual and
time-series sample moments to assign individual and time group
memberships. Cluster dummies are then interacted and enter the model
specification as group effects, and the associated parameters are
estimated along with the regression coefficients in the second
step. The resulting second-step estimator is consistent in the
presence of unspecified forms of the time-varying unobserved
heterogeneity with minimal assumptions on the unobserved components,
which makes it a perfect candidate to contrast with the OW-FE
estimator that is consistent with only time-constant (or
time-specific) heterogeneity. Note that, in order to perform the
proposed test, there is no need to estimate interactive fixed effects
models, as the TW-GFE encompasses this as well as more sophisticated
specifications for the unobserved heterogeneity.

We show that, under the null hypothesis of one-way effects, the TW-GFE
estimator converges to the OW-GFE, for which the asymptotic framework
has been established by \cite{blm2022}, implying that the Hausman
statistic \citep{Hausman1978} has asymptotic $\chi^2$
distribution. There are, however, two sources of asymptotic bias: the
incidental parameters problem, that in nonlinear models plagues both
estimators, and the approximation bias, that affects the TW-GFE. We
tackle the resulting non-centrality by reducing the bias of the vector
of constrasts via leave-one-out jackknife \citep{HN2004}. In addition,
we rely on parametric bootstrap \citep{mackinnon2006,
  horowitz2019bootstrap} to estimate the variance of the debiased
vector of contrasts.

In the linear model, our test can also be used to detect violations of
two-way additive individual- and time-specific heterogeneity, as
one-sided demeaning reconciles this setting with the specification of
one-way heterogeneity of our null hypothesis. In this framework,
additional assumptions on one of the heretogeneity components and a
modification of the moment functions used the the first-step
clustering are required. One potential application is the pre-trend
test to verify the crucial idenfying assumption of pre-treatment
common trends when applying Difference-in-Differences. In this
respect, the proposed test finds natural application in contexts where
the generating model for potential outcomes has an interactive fixed
effects structure \citep{callaway2023treatment, bai2024causal}.

We report the results of an extensive Monte Carlo study showing
evidence that the test has correct size and good power. Size
properties crucially depend on how effective is the clustering
procedure in approximating the unobserved heterogeneity for the
TW-GFE, that is choosing a sufficiently large number of groups and
ensuring that the moments used for the {\em kmeans} clustering are
informative about the latent traits and common factors. Power
properties are studied under the alternative hypothesis of a factor
structure.  While computationally more intensive than the testing
procedures put forward by \cite{castagnetti2015} and \cite{BBP2015},
the proposed test represents an improvement as the former can only be
applied to linear models in a large-$T$ framework and the latter,
while viable for generalized linear models admitting sufficient
statistics for incidental parameters, lacks power when time effects
are independent.

 We also provide two empirical applications for the
proposed test. The first concerns a linear model for the determinants
of housing prices in the U.S. : the test rejects the null
hypothesis of time-invariant heterogeneity, in line with  the literature suggesting
a factor structure for unobservable traits  \citep{holly2010spatio}.
 In the second application we analyze
the inter-temporal decisions of labor market participation for female
workers, revisiting the application in \cite{dj2015}, among others.  The test does not provide evidence of a more
complex structure for the unobserved heterogeneity, as it fails to
reject the null hypothesis of time-variant latent traits.

\subsection{Literature review}
This paper relates to the stream of
literature that has studied fixed-effects panel data models with
grouped structures for the unobserved heterogeneity. Discrete
heterogeneity has long been considered within the random-effects
approach \citep{Heckman1984}, especially by a large body of
statistical literature; see, for instance, \cite{MacLahlan2000} on
finite-mixture models and \cite{Bartolucci2012} on latent Markov
models. On the contrary, the investigation of grouped patterns of
heterogeneity in fixed-effects models is relatively recent in the
econometric literature.

\cite{hahnmoon2010} study the asymptotic bias arising form the
incidental parameters problem in nonlinear panel data models where
unobserved heterogeneity is assumed to be discrete with a finite
number of support points.  \cite{besterhansen2016} investigate the
asymptotic behavior of the ML estimator for nonlinear models with
grouped effects, under the assumption that subjects are clustered
according to some external known classification.  Models with unknown
grouped membership are studied by \cite{phillips2016}, who propose
penalized techniques for the estimation of models where regularization
by classifier-Lasso shrinks individual effects to group coefficients,
by \cite{andobai2016, ando2023large} who consider unobserved group
factor structures in (generalized) linear models with interactive
fixed effects, and finally by \cite{wang2023, lumsdaine2023estimation},
studying group structures combined with structural breaks.

Discrete unobserved heterogeneity can serve as a regularization device
that allows to identify the parameters of interest in panel data
models with time-varying individual effects but not necessarily
characterized by a factor structure. In this vein, \cite{bm2015}
introduce a GFE estimator for linear models where the discrete
heterogeneity is assumed to follow time-varying grouped patterns and
cluster membership is left unrestricted. By contrast, the TW-GFE
estimator by \cite{blm2022} is consistent even with unspecified forms
of time-varying unobserved heterogeneity. While using discretization
as an approximation device introduces an asymptotic bias, the function
of the unobserved heterogeneity they consider encompasses a variety of
specifications, such as additive and interactive effects, under
minimal distributional assumptions.  To the best of our knowledge, the
only alternative approach based on approximating heterogeneity is by 
\cite{freemanweidner2023}, which is however viable only for the linear model.

\subsection{Outline}
 The rest of the paper is organized as follows: Section
\ref{background} briefly describes the models and estimators; Section
\ref{sec:asym} reviews the assumptions required to
characterize the asymptotic distribution of the TW-GFE  and outlines the main theoretical contribution;
 Section \ref{specification tests} illustrates
the proposed test  and the asymptotic behavior of the resulting
test statistic, and briefly illustrates the alternative
testing procedures; Section \ref{simulation} presents the results of
the simulation study in both linear and probit cases; Section \ref{empirical} illustrates
the two empirical applications; Finally, Section \ref{conclusion}
concludes.

\section{Models, estimators, and intuition of the proposed test}
\label{background}

Consider a panel data setup where subjects are indexed by
$i=1, \ldots, N$ and time occasions are indexed by $t=1, \ldots,
T$. Throughout the paper, we assume that observations are independent,
conditional on the observed covariates and unobserved heterogeneity,
and that the models are {\em static}.
The traditional specification of OW-FE models depicts
unobserved heterogeneity as individual-specific intercepts, so that
the conditional distribution of the response variable $y_{it}$ given
an $r$-vector of exogenous covariates $x_{it}$ is of the type
\begin{equation}\label{eq:ti}
  y_{it}| x_{it},  \th_0, \alpha_{i0} \sim  f(y_{it} | x_{it}' \th_0 + \alpha_{i0}), \quad
\end{equation}
where $ \th_0 $ is the vector of parameters of interest, $\alpha_{i0}$
denotes the individual-specific time-invariant effect, and $f(\cdot)$
is a generic known density function. When \eqref{eq:ti} is a linear
regression model, consistent OLS estimators of $\th_0$ can be
trivially obtained on the basis of standard de-meaning or
first-differences transformations, whereas ML estimators in non-linear
models are not consistent in $T$ is fixed, unless probability
formulations admit sufficient statistics for the individual intercepts
\citep{and:70, Chamberlain1980}. Instead, under rectangular array
asymptotics \citep{Li2003}, the ML estimator is consistent but
exhibits a bias in the limiting distribution, which can be reduced by
analytical or jackknife corrections \citep{HN2004, AH2007}.

In this paper, we contrast the TW-GFE  with the OW-FE
estimator to perform specification tests and possibly detect
more sophisticated structures for the unobserved heterogeneity.
Consider the following model formulation
\begin{equation}\label{eq:tv}
  y_{it}| x_{it},  \th_0, \alpha_{it0}  \sim  f(y_{it} | 
  x_{it}' \th_0 + \alpha_{it0}).
\end{equation}
The time-varying unobserved heterogeneity $\alpha_{it0}$ is
characterized by two vectors $ \xi_{i0}$ and $ \lambda_{t0}$, and a
function
$\alpha$$(\cdot)$, satisfying requirements that will be discussed
later in more detail, such that $\alpha_{it0} = \alpha(
\xi_{i0},\lambda_{t0})$. This characterization of
$\alpha_{it0}$ follows that of \cite{blm2022} and can be easily
reconciled with the structures for the unobserved heterogeneity in
models \eqref{eq:ti} and \eqref{eq:tv}:
\begin{equation*}
  \alpha_{it0} : \left\{ \begin{array}{lc}
                           \alpha_{i0} \equiv \alpha(\xi_{i0})  & 
                                                                  \mathrm{in
                                                                  \quad
                                                                  \eqref{eq:ti}}\\[2pt]
                           \alpha_{it0} \equiv \alpha(\xi_{i0},
                           \lambda_{t0})
                                                                &
                                                                  \mathrm{in \quad \eqref{eq:tv}}\\
                         \end{array} \right.
                     \end{equation*}
                     where $\alpha_{i0}$ depends on individual traits only. 
                     It is worth noting that the formulation for
                     time-varying heterogeneity accomodates the
                     widespread approach based on including common
                     time effects that enter the specification in an
                     additive manner, i.e., $ \alpha_{i0} + \zeta_{t0}
                     \equiv \alpha(\xi_{i0},
                     \lambda_{t0})$ where
                     $\zeta_{t0}$ represents such time-varying
                     heterogeneity.

                     As for the TW-GFE approach, the estimator is
                     obtained via a two-steps procedure. The
                     first-step deals with the classification of
                     subjects and time occasions into two different
                     sets of groups. It is worth to stress that
                     clustering here serves as an approximation tool
                     for the unobserved heterogeneity, so that there
                     is no number of clusters to be known \emph{a
                       priori}. As a consequence, groups should not be
                     intended as aggregation levels coming from
                     external information \citep[e.g. sectors for
                     firms, see also][]{papke2023simple}.
                     Classification relies on performing {\em kmeans}
                     clustering twice, using the vectors of sample
                     moments
                     $ h_i = \frac{1}{T}\sum_{t=1}^T h( y_{it},
                     x_{it})$ and
                     $ w_t = \frac{1}{N}\sum_{i=1}^Nw( y_{it},
                     x_{it})$ of fixed dimensions. Both vectors have
                     to be \emph{informative} about $ \xi_{i0}$ and
                     $ \lambda_{t0}$, respectively, meaning that
                     $ \xi_{i0}$ can be uniquely recovered from $ h_i$
                     for large $T$ and $\lambda_{t0}$ can be uniquely
                     recovered from $w_t$ for large $N$.  The two {\em
                       kmeans} clustering procedures return a number
                     of $K$ groups for the subjects and a different
                     number of $L$ groups for the time occasions, from
                     which two sets of dummies identifying the related
                     group memberships are created.  In the second
                     step, cluster dummies for the cross-sectional and
                     time dimensions are then interacted and enter the
                     linear index of the model specified for the
                     response variable as $KL$ group fixed
                     effects. Estimation is then carried out by ML and
                     the resulting estimator is consistent for the
                     regression parameters, although asymptotically
                     biased due to a combination of incidental
                     parameters problem and approximation error, as
                     shown by \cite{blm2022}.

The characteristics of the OW-FE and TW-GFE estimators are such that
the Hausman principle can be invoked to perform specification
testing. In particular, with $N,T \rightarrow \infty$, the OW-FE
estimator is consistent under the null hypothesis of time-constant unobserved
heterogeneity, whereas the TW-GFE estimator remains consistent under
more complex time-varying forms of heterogeneity.

Let $\hat{\delta}$ be the difference between the OW-FE and TW-GFE
estimators. The null hypothesis of the proposed test is then
\[
H_0: \underset{N,T \rightarrow \infty}{\mathrm{plim}} \hat{\delta} = 0.
\]  
Under $H_0$ we show that the Hausman statistic
\[
H^{\dagger} = NT \hat{\delta}^{\dagger '}  \left(\widehat{W}_*\right)^{-1} \hat{\delta}^\dagger
\]
is asymptotically chi-square distributed with degrees of freedom equal
to the dimension of $\theta_0$,  where $\hat{\delta}^\dagger$ is the
debiased vector of contrasts  and $\widehat{W}_*$ is its bootstrap
variance estimator. This result relies on the main
theoretical contribution of this paper, which is the derivation of the
asymptotic distribution of the TW-GFE estimator under $H_0$. In
fact, \cite{blm2022} do not provide conditions under which the 
TW-GFE estimator can be debiased.

\section{Assumptions and asymptotic behavior of the compared
  estimators}\label{sec:asym}

In the following we recall the main assumptions and asymptotic results
for the OW-FE estimator, along with the characterization of the
asymptotic distribution of the TW-GFE estimator under $H_0$.  The
assumptions listed below recall those in \cite{blm2022, suppblm2022}.

\begin{ass}\label{ass:uh} Unobserved Heterogeneity: 
  There exist $ \xi_{i0}$ of fixed dimension
  $d_{\xi}$ and a function $\alpha(\cdot)$   and
  $ \lambda_{t0}$ of fixed dimension $d_\lambda$ and two functions
  $\alpha(\cdot)$ and $\mu(\cdot)$ that are Lipschitz-continuous in
  both arguments, such that
  $\al_{it0}=\al(\xi_{i0},\lambda_{t0})$ and
  $\mu_{it0}=\mu(\xi_{i0},\lambda_{t0})$;
(ii) the supports of $ \xi_{i0}$
  and $ \lambda_{t0}$ are compact.
\end{ass}
Assumption \ref{ass:uh} gives the minimal properties of the unobserved
heterogeneity. This specification encompasses time-constant
heterogeneity, i.e, $\alpha_{it0} \equiv \alpha_{i0}$, with
$\lambda_{t0} = \lambda_0$. In general, we will consider $d_{\xi} = 1$ and
$d_\lambda$ either $0$ or $1$. It is also important to the GFE approach that
covariates are affected by the same source of heterogeneity, so that
$x_{it}$ depends on $\mu_{it0}$, where
$\mu_{it0} = \mu( \xi_{i0}, \lambda_{t0})$, with $\mu(\cdot)$
satisfying the same requirements as $\alpha(\cdot)$.\footnote{The
  requirements on $\alpha(\cdot)$ and $\mu(\cdot)$ are those necessary
  for the viability of the GFE approach. Clearly the standard ML
  fixed-effects framework is unaffected by such assumptions.}

\begin{ass}\label{ass:dgp} Sampling: 
(i) $(y_{it}, x'_{it})'$, for $i=1,\dots N$ and $t=1,\dots T$, are i.i.d.
given $\xi_{i0}$ and $ \lambda_{t0}$;  (ii) $ \xi_{i0}$ and $ \lambda_{t0}$ are also i.i.d.
\end{ass}
Assumption \ref{ass:dgp}
outlines the sampling requirements that are more restrictive than that
usually required to characterize the asymptotic distribution of ML
estimators under rectangular-array asymptotics for fixed-effects
models with time heterogeneity. For example, \cite{FVW2016} assume
independence over $i$ while relaxing time independence by allowing for
$\alpha$-mixing.\footnote{See \cite{FVW2016}, Assumption 4.1 (ii).}
Assumption \ref{ass:dgp} is instead required for consistency of the
TW-GFE, which effectively rules out the possibility of applying the
proposed test to models with (i) feedback effects and (ii) unobserved
heterogeneity that depends on dynamic factors.

\begin{ass}\label{ass:reg} 
  Regularity: Let $\ell_{it}(\alpha_{it}, \th)$= ln $f (y_{it} | x_{it}, \alpha_{it}, \th)$  and let \\
  $\frac{1}{NT} \sum_{i=1}^N\sum_{t=1}^T
  \ell_{it}(\bar{\alpha}(\th,\xi_{i0},\lambda_{t0}),\th)$, with
  $\bar{\alpha}(\th,\xi,\lambda)
  =\argmax{\alpha}\mathbb{E}_{\xi_{i0}=\xi,\lambda_{t0}=\lambda}(\ell_{it}(\alpha,\th))$, 
  be the
  target log-likelihood:\\
  (i) $\ell_{it}(\alpha, \th)$ is three time differentiable in
  $(\alpha, \th)$; $\th_0$ is an interior point of the parameter
  space
  $\Theta$; $\Theta$ is compact;  \\
  (ii) $\ell_{it}$ is strictly concave as a function of $\alpha$,
  $\emph{inf}_{\xi,\lambda,\th}\mathbb{E}_{\xi_{i0}=\xi,\lambda_{t0}=\lambda}
  \left( -\frac{\pa^2
      \ell_{it}(\bar{\alpha}(\th,\xi,\lambda),\th)}{\pa\alpha\pa\alpha'}\right)
  >0;\\ \mathbb{E}[\frac{1}{NT} \sum_{i=1}^N\sum_{t=1}^T
  \ell_{it}(\bar{\alpha}(\th,\xi_{i0},\lambda_{t0}),\th)]$ has a
  unique maximum at $\th_0$ on $\Theta$, and its second derivative is
  negative definite.
  \\
  (iii) Regularity conditions on boundedness of moments and asymptotic
  covariances in \cite{suppblm2022} Assumption S2 (iv,v) apply. \\
  (iv) There exists a function $M(y_{it}, x_{it})$ such that $\left| \ell_{it}(\alpha_{it}, \th) \right| \leq M(y_{it}, x_{it})$;
  $\left|\frac{\pa \ell_{it}(\alpha_{it}, \th)}{\pa(\th,\al_{it})}\right| \leq M(y_{it}, x_{it})$; $\sup_i \mathbb{E} \left[ M(y_{it}, x_{it})^{33} \right] < \infty$; $ \sup_i \mathbb{E} \left[ M(y_{it}, x_{it})^{Q} \right]<\infty$, for some $Q>64$; $\left | \frac{\pa^{m1+ m2}\ell_{it}(\alpha_{it}, \th)}{\pa\th^{m1}\pa\alpha^{m2}}\right | \leq M(y_{it}, x_{it}) $ for $0,\leq,m1+m2,\leq 1, \dots, 6$\\
\end{ass}
The conditions stated
in Assumption \ref{ass:reg} are standard requirements for a well-posed
maximization problem.
Under Assumptions \ref{ass:uh},\ref{ass:dgp}, and \ref{ass:reg} the
OW-FE estimator of $ \theta$, $\hat{\th}$, for model \eqref{eq:ti} is
consistent as $N,T \to \infty$. For asymptotic normality, $N$ and $T$
are also required to grow at the same rate.
\begin{ass}\label{ass:rect} Asymptotics: as $N, T \rightarrow \infty$, $N/T \rightarrow \rho^2$,
  with $0 < \rho < \infty$.
\end{ass}
Additionally under Assumption
\ref{ass:rect}, $\hat{\theta}$ has the following asymptotic
distribution
\begin{equation*}
  \sqrt{NT}(\hat{ \th} - \th_0) \convd N \left( \frac{B}{T},  \,\, I(\th_0)^{-1})\right),
\end{equation*}  
where $ B$ is constant. Notice that, in the case of informational
orthogonality between the structural and nuisance parameters, such as
in the linear model, $B = 0$, whereas characterizations of these
asymptotic biases are given in \cite{HN2004} for nonlinear models with
OW-FE. Finally, $I(\th_0)$ is the information matrix of the profile
log-likelihood.

Consistency of the TW-GFE estimator relies on
specific conditions ensuring that the first-step discretization
effectively approximates the unobserved
heterogeneity. These conditions have to be placed on the information
used for clustering, namely sample moments of $(y_{it}, x_{it})$.
\begin{ass}\label{ass:mom}  Moment informativeness:
  There exist sample moment vectors
  \[h_i = \frac{1}{T}\sum_{t=1}^T h(y_{it},x_{it}'), \quad
     \quad
    w_t= \frac{1}{N}\sum_{i=1}^N w(y_{it},x_{it}')
    \]
    of fixed dimension, and two unknown
    Lipschitz-continuous functions $\phi$ and $\psi$, such that
    \[
      \underset{T \to \infty}{\mathrm{plim}} \, h_i = \phi(\xi_{i0}),
      \quad 
      \quad \underset{N \to \infty}{\mathrm{plim}} \, w_t = \psi(\lambda_{t0}),
      \]
      and
  $ \frac{1}{N}  \sum_{i=1}^N \Vert h_i - \phi(\xi_{i0})\Vert^2 =
  O_p(1/T)$,
  $ \frac{1}{T} \sum_{i=t}^T  \Vert w_t - \psi(\lambda_{t0})\Vert^2 =
  O_p(1/N)$ as $N,T \to \infty$. Furthermore, there exist two 
    Lipschitz-continuous functions $\varphi$ and $\nu$, such that
    $\xi_{i0} = \varphi(\phi(\xi_{i0}))$ and $\lambda_{t0} = \nu(\psi(\lambda_{t0}))$.
\end{ass}
Assumption \ref{ass:mom} formalizes moments informativeness in order
to have an effective individual and time clustering. Intuitively, by
comparing two different sets of moments $h_i$ and $h_j$, two different
types $\xi_{i0}$ and $\xi_{j0}$ can be separated. This is guaranteed
by sample moments being asymptotically injective functions of the
unobserved heterogeneity. 

Let $\tilde{\theta}$ be the TW-GFE estimator. Then under Assumptions 
\ref{ass:uh}, \ref{ass:dgp}, \ref{ass:reg}, and \ref{ass:mom},
\cite{blm2022} show that the TW-GFE estimator has asymptotic expansion
 \begin{equation}\label{eq:gfe}
\hspace*{-6cm}\tilde{\theta} = \theta_0 + J(\th_0)^{-1}\frac{1}{NT}\sum_{i=1}^{N}\sum_{t=1}^{T}
s_{it}(\th_0)
\end{equation}
\[
\hspace*{2cm} + O_p\left(\frac{1}{T}+ \frac{1}{N} + \frac{KL}{NT}\right) 
+ O_p(K^{-\frac{2}{d_\xi}} + L^{-\frac{2}{d_\lambda}} )
+ o_p\left(\frac{1}{\sqrt{NT}}\right),
\]
as $N,T,K,L \rightarrow \infty$, such that $KL/(NT)$ tends to zero
\citep[cf. Corollary \S S2][]{suppblm2022}. Here $J(\cdot)$ and
$s_{it}(\cdot)$ are the negative expected Hessian and the score
associated with the likelihood function.  Three main different sources
of bias can be identified: the $1/T$ and $1/N$ terms depend on the
number of time occasions and subjects used for $h_i$ and $ w_t$ in the
classification step; the $KL/NT$ term reflects the estimation of $KL$
group-specific parameters using $NT$ observations; the
$K^{-\frac{2}{d_{\xi}}} + L^{-\frac{2}{d_{\lambda}}}$ terms refer to
the approximation bias arising from the discretization of $ \xi_{i0}$
and $ \lambda_{t0}$ via {\em kmeans}.

The $O_p(\cdot)$ terms in the above expansion can be shown to become
$O_p(1/T + 1/N)$ under suitable choices for the number of groups, $K$
and $L$. The rule suggested by \cite{blm2022} and the consequent
characterization of the $O_p(\cdot)$ terms are summarized in the
following proposition.

\begin{remark}\label{remark:kandl} Number of groups and approximation bias:\\
  i) The number of groups $K$ and $L$ are
    chosen according to the following rules
    \begin{equation*} 
      \hat{K} = \min_{K\ge 1}\{K:\hat{Q}(K) \le \gamma_K\hat{V}_{h}\},
\qquad \hat{L} = \min_{L\ge 1}\{L:\hat{Q}(L) \le \gamma_L\hat{V}_{w}\},
\end{equation*}
where $Q(\cdot)$ is the  {\em kmeans} objective function, $\gamma \in (0,1]$ is a user-specified
parameter, $ \hat{V}_h = \mathbb{E}\left[ \Vert h_i - \phi(\xi_{i0})  \Vert^2\right]  +
 o_p\left( \frac{1}{T} \right)$, and $ \hat{V}_w = \mathbb{E}\left[ \Vert w_t - \psi(\lambda_{t0})  \Vert^2\right]  +
 o_p\left( \frac{1}{N} \right)$.\\
ii) With $K = \hat{K}$ and $L = \hat{L}$,
the approximation errors are $O_p(1/T)$ and $O_p(1/N)$, so that the
$O_p(\cdot)$ terms in \eqref{eq:gfe} become $O_p(1/T + 1/N)$.
\end{remark}
\noindent
The quantities $V_h$ and $V_w$ are the variability of the moments $h_i$ and $ w_t$,
that, under Assumption \ref{ass:dgp}, can be estimated by
$\hat{V}_h = \frac{1}{NT^2}\sum_{i=1}^N\sum_{t=1}^T \Vert h(y_{it},
 x_{it}') - h_i  \Vert^2$ and $\hat{V}_w = \frac{1}{N^2T}\sum_{i=1}^N\sum_{t=1}^T \Vert w(y_{it},
 x_{it}') - w_t  \Vert ^2$, respectively. We refer the reader to
 \cite{blm2022} for derivations and proofs of these results. In practice, the value $\gamma$ set by
the researcher governs the number of groups, with  smaller values of $\gamma$ yielding a
larger number of groups.

We now state our main theoretical result, that is the asymptotic
distribution of the TW-GFE under the null hypothesis of time-constant
unobserved heterogeneity which, as per Assumption \ref{ass:mom}, can
now be formulated  $H_0: \lambda_{t0} =  \lambda_{0}$, for $t=1, \ldots,T$. 
\begin{theorem}\label{th:tw-gfe}
  Under Assumptions \ref{ass:uh}-\ref{ass:mom}, when data are sampled
  under $H_0$ and when the number of groups are determined according
  to the rule in Remark \ref{remark:kandl} (i), then
  \begin{equation}
    \sqrt{NT}(\tilde{\th} - \theta_0) \convd N\left(\frac{C}{T},  J(\th_0)^{-1}\right),
  \end{equation}
  where the bias term $C$ is given by \cite{blm2022}, Corollary \S 2. 
\end{theorem}
\noindent
The proof is given in Appendix \ref{p_th_tw_gfe}. The TW-GFE is therefore asymptotically normal with constant bias as
$N,T$ grow at the same rate. Hence, standard theory on the Hausman
test can be applied and the non-centrality of its distribution,
arising from the asymptotic biases of both the OW-FE and TW-GFE
estimators, can be addressed by conventional bias reduction techniques.

\section{Specification test}
\label{specification tests}

In this section, we first illustrate the proposed test stating the related theoretical results; we then offer an overview of the alternative approaches already present in the literature; finally, we describe a potential additional application for the proposed procedure to the testing for parallel trends in linear two-way fixed-effects models.  

\subsection{Proposed test}
We propose a Hausman test for the specification of the unobserved
heterogeneity considering, as null hypothesis $H_0$, that data are
generated from the model portrayed by Equations \eqref{eq:ti}.  In
order to perform the test we rely on a Hausman-type statistic based on
the difference $\hat{\delta} = \hat{\th} - \tilde{\th}$, namely, by
contrasting OW-FE estimator with the TW-GFE estimator.
Formally, the null hypothesis can be expressed as
\[
H_0: \underset{N,T \rightarrow \infty}{\mathrm{plim}}\hat{\delta} = 0.
\]

We cannot adopt the traditional formulation of the
Hausman test \citep{Hausman1978}, for which the test statistics is
\begin{equation}\label{eq:H_asym}
 \hat{H} = NT\hat{\delta}' \widehat{W}^{-1}\hat{\delta} ,
\end{equation}
where $\widehat{W}$ is a consistent estimator of the variance of the
contrasts $\hat{\delta}$, $W$. In our scenario, the asymptotic biases
of $\tilde{\th}$ and $\tilde{\th}$ propagate to the asymptotic
distribution of $\hat{\delta}$ under $H_0$.
\begin{cor}\label{cor:delta}
  Assume that $V(\hat{\delta}) = W_0$ exists and is non-singular and
  that Assumptions \ref{ass:uh}-\ref{ass:mom} hold. If data are sampled
  under $H_0$ and the number of groups are determined according
  to the rule in Remark \ref{remark:kandl} (i), then
  \begin{equation}
    \sqrt{NT} \; \hat{\delta} \convd N\left(\frac{C -B}{T}, W_0 \right),
  \end{equation}
\end{cor}
Corollary \ref{cor:delta} shows that the bias in the asymptotic
distribution of $\hat{\delta}$ makes $\hat{H}$ a non-pivotal
quantity. The proof follows by standard arguments. In order to
overcome this issue, we remove the $O(1/T)$ bias term of $\hat{\delta}$ via the the
leave-one-period-out jackknife procedure \citep{HN2004}. We therefore
consider
\begin{equation}\label{eq:delta_jk}
  \hat{\delta}^{\dagger} = T \cdot \hat{\delta} - \frac{(T-1)}{T} \cdot \sum_{t = 1}^T \hat{\delta}^{(t)},
\end{equation}
where $\hat{\delta}^{(t)}$ denotes the contrasts vector computed on a
subsample excluding observations in period $t$.
The following theorem established asymptotic validity of the jackknife
bias reduction for the vector of contrasts.
\begin{theorem}\label{th:jack_delta}
  Under the assumptions of Corollary \ref{cor:delta}, we have
  \begin{equation}
    \sqrt{NT} \; \hat{\delta}^\dagger \convd N\left(0, W_0 \right).
  \end{equation}
  \begin{proof}
The proof is given in Appendix \ref{th_2_j}. 
  \end{proof}
\end{theorem}

The jackknife correction poses an additional challenge concerning the estimation of the covariance matrix of $\delta^{\dagger}$, say
$\widehat{W}$. First, the analytical expression should take into account the two-steps nature of the TW-GFE estimator. Secondly, it is not straightforward to derive the formulation for the variance of the difference of the two estimators. Finally, the estimator may not reflect the potential variance inflation due to the jackknife procedure in finite samples.\footnote{For example, \cite{FVW2016} show how the
  split-panel jackknife correction induces a variance inflation of the
  ML estimator of regression parameters in the two-way additive fixed
  effects models.} We overcome this issue by relying on a bootstrap
estimate of the variance of the contrasts.\footnote{\cite{bartolucci2023conditional}
  also consider bootstrap standard errors for marginal effects of
  fixed-effects logit models after a split-panel jackknife
  correction.}
Specifically we rely on parametric bootstrap, thus generating the data based on the OW-FE estimates
$\hat{\theta}$, that reflect the model specification under $H_0$. We compute the quantity in Equation
\eqref{eq:delta_jk} for each of the $R$ generated samples in order to get the bootstrap estimate of its covariance matrix, denoted as $\widehat{W}_{*}$.

The proposed test statistic is therefore defined as a Hausman
statistic, as outlined by the Corollary below, the proof of which
follows by standard arguments.
\begin{cor}\label{cor:test}
 Under the assumptions of Corollary \ref{cor:delta}, we have
  \[
     \hat{H}^\dagger = NT \hat{\delta}^{\dagger \prime} \left( \widehat{W}_{*} \right)^{-1}\hat{\delta}^{\dagger} \convd \chi^2_r.
   \]
\end{cor}


\subsection{Alternative approaches} \label{ALTRITEST}
There are two
alternative tools that can be used to test the same null hypothesis as the one here considered:  the max-type test put forward by \cite{castagnetti2015}
(CRT test henceforth) to detect factor structures in a linear
framework and the test for time-invariant unobserved heterogeneity
developed by \cite{BBP2015} (BBP test) that applies to both
linear and non-linear models.

The CRT test considers a linear model in which the unobserved heterogeneity is depicted as in \eqref{eq:tv}
and $\alpha_{it0}$ has a factor structure, that is
$\alpha(\xi_{i0}, \lambda_{t0}) = \alpha_{i0}'\zeta_{t0}$. 
The procedure tests the null hypothesis of no factor structure,
defined as $H_0: \zeta_{t}= \zeta$, that is a model with only individual effects.
The max-type test statistics is formulated as 
  \begin{equation*}
    S = \text{max}_{1\leq t \leq T} \left[N(\hat{\zeta}_t - \hat{\bar{\zeta}})'\hat{\Sigma}_{ t}^{-1}(\hat{\zeta}_t - \hat{\bar{\zeta}})\right],
  \end{equation*}
  where factors are estimated using the common correlated effects
  approach by \cite{pesaran2006},\footnote{It is worth recalling that
    the approach by \cite{castagnetti2015} can in general be
    implemented in models with heterogeneous slopes.}
  $\hat{\bar{\zeta}}$ is the sample mean of $\hat{\zeta}_t$ and the
  $\hat{\Sigma}_{t}$ is an estimate of the asymptotic factor
  covariance matrix \citep[crf Equation 10 in][]{castagnetti2015}. The
  test statistic $S$ has an asymptotic Gumbel distribution. It is
  worth noting that CRT test requires large $T$ settings in order
  attain the correct size in finite samples.
  
Differently from CRT test, the BBP test can be employed with
generalized linear models that admit a sufficient statistic for the OW unobserved heterogeneity parameters. Therefore, the BBP
is a generalized Hausman test that compares estimators that are
consistent only under time-constant unobserved heterogeneity with
estimators that are consistent even in presence of a time-varying
latent variable. 

For linear models, the BBP test contrasts the OW-FE estimator with the first difference estimator. For the logit and Poisson models, the BBP test compares estimators based on two different formulation of the conditional likelihood: the standard conditional ML (CML) estimator and the Pairwise CML estimator (PCML). The former is consistent with time-constant heterogeneity since the probabilities for the models considered admit sufficient statistics for the incidental parameters. The PCML approach considers pairs of consecutive time observations for every individual and the corresponding log-likelihood is conditioned on the sufficient statistic, thereby allowing for different individual effects in every couple of periods. The PMCL estimator is therefore consistent in presence of time-varying unobserved heterogeneity. It is worth stressing that the BBP test has power only when certain conditions are met, namely $T$ must be greater than 3, otherwise the estimators coincides, and common factors must have a dynamic structure. 

\subsection{Testing for parallel trends}
In the following we illustrate the application of the proposed procedure to test for parallel trend assumption in a Difference-in-Differences (DiD) setting, which is crucial to achieve identification of the causal effect of interest. 
\footnote{DiD one of the most widely applied methods in treatment evaluation and the related literature is indeed vast. We refer to \cite{roth2023s} for a review about recent developments and open issues.
}


The standard setting in empirical applications is based on a linear model including additive individual and time heterogeneity, that is the Two-Way FE (TW-FE) model \citep{de2023two,de2024difference}. This formulation implies homogeneity of the time effects across units and time-constant individual heterogeneity, that are crucial for the parallel trend assumption to hold but might be of concern in applications where individual latent variables are likely to change over time. For these reasons, recent contributions have focused on departures from the additive specification for the fixed effects. In particular,
\cite{callaway2023treatment} discuss the identification and estimation
of the treatment effect on the treated relying on an interactive
effects model for the untreated units, allowing for heterogeneous
trends among units. Similarly, \cite{gobillon2016regional} first, and
\cite{bai2024causal} later, extend the latter approach in order to
model both treated and untreated units.

The proposed method can be used to test for the form of the
unobserved heterogeneity  in the following  model
\begin{equation}
  \label{eq:tw-did}
  Y_{it}(0) = \xi_i + \lambda_t + \eta_if_t + x_{it}'\theta + \varepsilon_{it},
\end{equation}
where the parallel trend assumption is violated  by  the presence, in the pre-treatment period, of heterogeneous time effects, namely $\eta_if_t$.

As the proposed test can be used to detect departures from  specifications with  one-way unobserved heterogeneity, the quantities in Equation \eqref{eq:tw-did} must be demeaned in order to contemplate one source of heterogeneity under the null hypothesis. Since the aim is to verify the parallel trend assumption, we transform the model by eliminating time-constant heterogeneity, so that Equation \eqref{eq:tw-did} becomes
\begin{equation}
  \label{eq:tw-did_dem}
  \ddot{Y}_{it}(0) =   \ddot{\lambda}_t + \eta_i\ddot{f}_t + \ddot{x}_{it}'\theta + \ddot{\varepsilon}_{it},
\end{equation}
 where the superscript refers to the individual de-meaned variable, e.g $\ddot{x}_{it}= x_{it}  - \sum_{t=1}^{T} x_{it}$. 
   In this vein, the proposed test compares the TW-GFE and the OW-FE estimator for a model where the one-way unobserved heterogeneity is represented by time fixed effects. Both  are consistent under the null hypothesis of the presence of a common trend, that is $H_0:\eta_i=\eta$. 
   
 There are  two key differences with respect to the aforementioned use of our test. 
Due to the individual de-meaning, the moments used in the search of the $K$ groups are null vectors by construction and hence they would not lead to meaningful clustering. We propose to tackle this issue by employing second moments to discretize the individual heterogeneity. This choice does not conflict with the requirements on  the  unobserved heterogeneity, provided the following assumption holds.
\begin{ass}\label{ass:possupp} Positive support:
The support of \( \eta_i \) is a bounded subset of \( \mathbb{R}_+ \), i.e., \( \eta_i \in [0, a] \) for some \( a > 0 \).
\end{ass}
Assumption \ref{ass:possupp} guarantees the informativeness of individual-specific  second moments as it ensures the injectivity and Lipschitz-continuity of the function  $h(z_{it})=1/T\sum_{t=1}^Tz_{it}^2$, as required by Assumption \ref{ass:mom}. Instead, clustering for time occasions does not pose any  additional issue. 

Furthermore, since we employ OW-FE with only time effects,
we perform the leave-one-individual-out jackknife correction on the vector of contrasts, as  the incidental parameter problem pertains to the time dimension. In \emph{formulae}:
\begin{equation*}\label{eq:delta_jk_for_i}
  \hat{\delta}^{\dagger} = N \cdot \hat{\delta} - \frac{(N-1)}{N} \cdot \sum_{i = 1}^N \hat{\delta}^{(i)},
\end{equation*}
where $\hat{\delta}^{(i)}$ denotes the  vector of contrasts computed on the
subsample excluding unit $i$. 
  
\section{Simulation study}
\label{simulation}

In the following Section, we describe the design and report the results of an
extensive simulation study. First, we investigate empirical
size and power properties of the test for the linear and probit model. We then turn to the comparison of the proposed approach with the CRT and BBP tests. Finally, we present simulation results in the context of testing for the parallel trend assumption in the TW-FE model.

\subsection{Linear model}

We design a Monte Carlo experiment where observations are
generated by a linear regression model with two exogenous covariates. 
 For the null hypothesis we consider a scenario where the
unobserved heterogeneity is specified as in 
model \eqref{eq:ti}. 

In particular, in the case of  individual time-constant effects, for
$i = 1, \ldots, N$ and $t=1, \ldots, T$, we generate samples according
to the following equations, which we denote as DGP-FE-L:
\begin{eqnarray*}\label{eq:dgp1}
 y_{it}& = & x_{it1}\th_1 + x_{it2}\th_2 + \alpha_i +
             \varepsilon_{it}, \\
  x_{itj} & = & \Gamma_{i} + N(0,1), \quad \mathrm{for} \quad j=1,2,
                \nonumber
\end{eqnarray*}
where $\alpha_i = \varrho\Gamma_i +\sqrt{(1-\varrho^2)}A_i$ , with
$A_i,\Gamma_i \sim N(1,3)$, and $\varrho=0.5$. Finally,
$\varepsilon_{it}$ is an idiosyncratic standard normal error term. We
let the coefficients $\th = (\th_1,   \th_2)'$ be equal to $(1, 1)'$. In this
design, we explore the size properties of the proposed test by comparing
the OW-FE with the TW-GFE estimator. 

In order to investigate the power of the proposed test, the scenario
generated under the alternative hypothesis is a linear panel data
model with interactive fixed effects. Specifically, samples are
generated according to a simplified version of the design outlined by
\cite{Bai2009}, with one latent factor:
\begin{eqnarray*}\label{eq:dgp3}
 y_{it}& = & x_{it1}\th_1 + x_{it2}\th_2 + \alpha_i\zeta_t + 
             \varepsilon_{it}, \\
  x_{it} & =  & \Gamma_i\zeta_t + N(0,1). \nonumber
\end{eqnarray*}
where $\zeta_t \sim N(1,\sigma^2_{\zeta_t})$. We consider
three values  of the standard deviation of the time factor, $\sigma_{\zeta_t}=(1,3,5)$.
We denote this design as DGP-IFE-L.  

For each scenario, we consider $N = 50,100$, $T = 10, 20$, and $199$
bootstrap replications  when computing the variance-covariance matrix of the vector of contrasts,
  for each of the $1000$ Monte Carlo replications.  It
is worth recalling that the performance of the TW-GFE estimator is
closely linked to the number of groups chosen for the first-step {\em
  kmeans} clustering. Even following the rule outlined in Remark
\ref{remark:kandl}, this number depends on the variability in the data,
which affects how informative $h_i, w_t$ are about the unobserved
heterogeneity, and the user-defined parameter $\gamma$. We account for
the former by allowing for large variances in the composite error
term, and for the latter by running scenarios where
$\gamma=0.2,0.5,1$, resulting in a decreasing number of clusters.

Table \ref{tab:sim_res_lin_FE} reports the average of the Hausman test $\hat{H}$ in
\eqref{eq:H_asym} across simulations, for both uncentered and centered
test statistics, along with the  respective empirical size
based on the quantile of the central $\chi^2_r$ with $r=2$.
When evaluating the null hypothesis, we report the empirical size for three
values of statistical significance (0.05, 0.1, 0.2) in order to show that the proposed
test exhibits a uniform p-value.
In Appendix \ref{tabelle_solo_stimatori}, we also present the average bias
and standard deviation for all estimators involved, under both the null and the alternative
hypotheses. For the TW-GFE estimator we also report the average selected number of groups in the first step,
according to the rule outlined in Remark \ref{remark:kandl}.

As expected, the empirical size of the uncentered test 
statistics  does not attain the nominal one.
 Instead, the centered Hausman test exhibits an empirical size that is always
close to the nominal one.
The $\ga$
parameter seems to have no effect on the size of the test: coherently with the theory, the number of groups found by
TW-GFE for the time dimension is always equal to one (see Table \ref{tab:sim_res_lin_FE_stim} in Appendix  \ref{tabelle_solo_stimatori}), meaning that the estimator manages to identify the absence of variation of heterogeneity over time.
 
 Table \ref{tab:sim_res_lin_FE_H1} reports results under the alternative hypothesis. The proposed test
 presents good rejection rates when the true DGP has a factor structure for the unobserved
heterogeneity, with power improving as both the standard deviation of the 
time factor and the $N,T$ dimensions increase. Accordingly, under the alternative hypothesis TW-GFE and its jackknife counterpart  exhibit a bias which
is substantially smaller with respect to ML and its bias corrected versions.
Moreover, we conjecture that further increasing the number of individuals/time occasions or enriching the
information entailed in factors may further improve the power of the test.

\begin{landscape} 
\begin{table}
  \caption{Size analysis: DGP-FE-L, OW-FE vs TW-GFE}
  \label{tab:sim_res_lin_FE}
  \begin{center}
    \begin{footnotesize}
      \begin{tabular}{lcccccccccccc}
        \hline
        \\
 &   \multicolumn{4}{c}{$\gamma = 0.2$} & \multicolumn{4}{c}{$\gamma = 0.5$} & \multicolumn{4}{c}{$\gamma = 1$} \\[2pt]
&  \multicolumn{2}{c}{T=10} &  \multicolumn{2}{c}{T=20} &  \multicolumn{2}{c}{T=10}&  \multicolumn{2}{c}{T=20}&  \multicolumn{2}{c}{T=10}&  \multicolumn{2}{c}{T=20}\\[2pt]
$N$ & \multicolumn{2}{c}{$50$\quad$100$} & \multicolumn{2}{c}{$50$\quad$100$} & \multicolumn{2}{c}{$50$\quad$100$} & \multicolumn{2}{c}{$50$\quad$100$} & \multicolumn{2}{c}{$50$\quad$100$} & \multicolumn{2}{c}{$50$\quad$100$}\\[2pt]
\hline \\  
Unc. H & 5.195 & 8.477 & 4.361 & 6.897 & 7.434 & 12.769 & 6.247 & 10.845 & 10.106 & 17.401 & 8.544 & 14.436 \\ [3 pt]
  size 0.05 & 0.323 & 0.584 & 0.266 & 0.459 & 0.511 & 0.799 & 0.395 & 0.717 & 0.654 & 0.924 & 0.589 & 0.839 \\ [3 pt]
  size 0.1 & 0.422 & 0.686 & 0.358 & 0.591 & 0.623 & 0.872 & 0.506 & 0.818 & 0.749 & 0.950 & 0.696 & 0.902 \\ [3 pt]
  size 0.2 & 0.557 & 0.798 & 0.475 & 0.732 & 0.728 & 0.940 & 0.643 & 0.884 & 0.841 & 0.977 & 0.801 & 0.951 \\ [3 pt]
 \hline \\  
  JK H & 2.020 & 2.191 & 1.904 & 1.964 & 2.102 & 2.095 & 1.987 & 2.007 & 2.083 & 2.131 & 2.105 & 2.067 \\ [3 pt]
  size 0.05  & 0.052 & 0.062 & 0.042 & 0.048 & 0.050 & 0.053 & 0.051 & 0.050 & 0.051 & 0.059 & 0.060 & 0.052 \\ [3 pt]
  size 0.1  & 0.103 & 0.116 & 0.082 & 0.093 & 0.110 & 0.100 & 0.091 & 0.101 & 0.105 & 0.112 & 0.100 & 0.111 \\ [3 pt]
  size 0.2 & 0.191 & 0.225 & 0.179 & 0.178 & 0.210 & 0.221 & 0.188 & 0.210 & 0.209 & 0.214 & 0.204 & 0.220 \\ [3 pt]
  \hline \\
  \end{tabular}
  \end{footnotesize}
  \end{center}
\begin{scriptsize}
  1000 Monte Carlo (MC) replications. ``Unc H'' is the average of the uncentered
  Hausman test statistic, across MC replications. `` JK H ''  is the average of the centered
  Hausman test statistic, across MC replications.
  ``size'' denotes the
  rejection rate for a nominal size of 5,10 and 20\%. 
 199 bootstrap  replications used for the jackknife correction.  
  \end{scriptsize}  
\end{table}
\end{landscape}


\begin{table}
  \caption{Power analysis: DGP-IFE-L, OW-FE vs TW-GFE}
  \label{tab:sim_res_lin_FE_H1}
  \begin{center}
      \begin{tabular}{llcccc}
        \hline
        \\
 &&   \multicolumn{4}{c}{$\gamma = 0.2$}\\[2pt] 
&&  \multicolumn{2}{c}{T=10} &  \multicolumn{2}{c}{T=20}\\[2pt] 
$N$& & \multicolumn{2}{c}{$50$\quad$100$} &
                                            \multicolumn{2}{c}{$50$\quad$100$} \\[2pt] 
\hline \\
\multirow{4}*{$\sigma_{\zeta_t}=1$}  \\ [ 2pt] 
&Unc. H & 42.579 & 89.786 & 49.375 & 106.501 \\  [ 2pt] 
  &Rejection rate & 0.976 & 0.999 & 0.987 & 0.999 \\  [ 2pt] 
  &JK H & 6.186 & 7.514 & 5.322 & 5.611 \\ [ 2pt] 
 & Rejection rate & 0.301 & 0.346 & 0.268 & 0.276 \\ [ 2pt] 
  \multirow{4}*{$\sigma_{\zeta_t}=3$}  \\ [ 2pt] 
  &Unc. H  & 45.736 & 93.446 & 90.061 & 197.193 \\  [ 2pt] 
  &Rejection rate & 0.828 & 0.894 & 0.914 & 0.944 \\ [ 2pt] 
  &JK H & 15.093 & 21.841 & 23.336 & 34.225 \\ [ 2pt] 
  &Rejection rate & 0.573 & 0.672 & 0.642 & 0.758 \\ [ 2pt] 
  \multirow{4}*{$\sigma_{\zeta_t}=5$}  \\ [ 2pt] 
  &Unc. H  & 39.195 & 81.273 & 83.788 & 171.059 \\ [ 2pt] 
  &Rejection rate & 0.772 & 0.861 & 0.847 & 0.916 \\ [ 2pt] 
  &JK H & 16.383 & 25.337 & 26.598 & 40.769 \\ [ 2pt] 
  &Rejection rate & 0.600 & 0.724 & 0.686 & 0.786 \\ [ 2pt] 
 \hline \\
 \end{tabular}
 \end{center}
 \begin{scriptsize}
 1000 Monte Carlo (MC) replications. ``Unc H'' is the average of the uncentered
  Hausman test statistic, across MC replications. `` JK H ''  is the average of the centered
  Hausman test statistic, across MC replications.
  ``Rejection rate'' denotes the
  rejection rate for a nominal size of 5\%. ``$\sigma_{\zeta_t}$'' refers to the standard deviation of the 
  time factor.
199 replications used for the jackknife correction.  
  \end{scriptsize}  
\end{table}

 \subsection{Probit model}
 We investigate the small sample properties of the test in a
 nonlinear setting, specifically by considering the probit model. 
 For $i = 1, \ldots, N$ and $t=1, \ldots, T$,
 we generate samples according to the following equations, which we
 denote as DGP-FE-P:
 \begin{eqnarray*}\label{eq:dgp1nl}
  y_{it}&= &\Indic\left(x_{it1}\th_1 + x_{it2}\th_2 +  \alpha_i + \varepsilon_{it} \geq0\right),  \\
  x_{its}&=&\sqrt{(1/5)}\left[\Gamma_{i} + N(0,1)\right], \quad \text{for}\,\, s=1,2,\nonumber
   \end{eqnarray*}
   where $\Indic(\cdot)$ is an indicator function, $\alpha_i = \varrho\Gamma_{i} +\sqrt{(1-\varrho^2)}A_i$,
   $A_i,\Gamma_{i} \sim N(1,1)$, $\varrho=0.5$, and 
   $\varepsilon_{it}$ is an idiosyncratic standard normal error
   term. The slope parameters $\th=[\th_1, \th_2]$ are both equal to 1.
   
  We evaluate the power of the proposed test  under
 the alternative hypothesis of interactive fixed effects with one latent
 factor, that is
 \begin{eqnarray*}\label{eq:dgp3nl}
   y_{it} & =  & \Indic\{x_{it1}\th_1 + x_{it2}\th_2 +  \alpha_i\zeta_t+ \varepsilon_{it} \geq0\},\\
   x_{it} & =  & \kappa[\Gamma_i\zeta_t + N(0,1)], \nonumber
 \end{eqnarray*}
 where  the product $\alpha_i\zeta_t$ is rescaled to have unit
 variance and  $\zeta_t \sim N(1,\sigma^2_{\zeta_t})$.  We refer to this design as DGP-IFE-P.

  For each experiment, we  consider $N = 50,100$, $T = 10,20$, and $199$ bootstrap draws when computing the variance-covariance
 matrix of the vector of contrasts,  in $1000$ Monte Carlo replications.
  When dealing with DGP-IFE-P,  we employ three values of the standard deviation of time factor $\sigma_{\zeta_t}=\sqrt{0.5},1,\sqrt{2}$. We focus on $\ga=1$ only.
  It is useful to recall that population moments used to cluster the unobserved heterogeneity may include the average of the binary dependent variable. In this case, however, the individual averages of $y_{it}$  may not
 provide enough information to detect more complex forms of unobserved
 heterogeneity, introducing only noise instead. For this reason we exclude the sample moments of the dependent variable in the clustering step.

 Tables \ref{tab:sim_res_nonlin_FE} and \ref{tab:sim_res_bin_FE_H1} report size and power analysis. Table \ref{tab:sim_res_bin_FE_stim} and \ref{tab:sim_res_bin_FE_H1_stim} in Appendix \ref{tabelle_solo_stimatori} report average bias and standard deviation of the estimators.
In the size analysis we report the empirical rejection rate for three
values of statistical significance (0.05, 0.1, 0.2).
As expected, the centered  Hausman test approaches the correct size and has uniform p-value, while
 the one based on biased estimators is not reliable. Unlike what happens with the linear model, here the test starts approaching the  nominal size when the $T$ dimension grows larger, namely from $T=20$ onward. However, due to the more intense computational cost in the nonlinear case, we restrict our analysis to smaller samples. 

 The power analysis shows that the test has also good power properties, increasing both in the $N,T$ dimensions and in the variance of the time factor. We conjecture that further increasing the sample size could lead to a larger rejection rate.

\begin{table}
  \caption{Size analysis, Probit model: DGP-FE-P, OW-FE vs TW-GFE}
  \label{tab:sim_res_nonlin_FE}
  \begin{center}
      \begin{tabular}{llcccc}
        \hline
        \\
 &&   \multicolumn{4}{c}{$\gamma = 1$}\\[2pt] 
&&  \multicolumn{2}{c}{T=10} &  \multicolumn{2}{c}{T=20}\\[2pt] 
$N$ && \multicolumn{2}{c}{$50$\quad$100$} &
                                            \multicolumn{2}{c}{$50$\quad$100$} \\[2pt] 
                                            \hline \\
Unc H && 7.074 & 15.125 & 12.944 & 27.297   \\ [2pt]
  size 0.05 && 0.595 & 0.989 & 0.890 & 0.999   \\ [2pt]
    size 0.1 && 0.783 & 0.997 & 0.949 & 1.000   \\ [2pt]
      size 0.2 && 0.923 & 1.000 & 0.973 & 1.000   \\[2pt]
      \hline \\
  JK H && 1.655 & 1.675 & 2.126 & 1.932   \\ [2pt]
  size 0.05 && 0.033 & 0.037 & 0.068 & 0.048   \\ [2pt]
  size 0.1 && 0.058 & 0.063 & 0.111 & 0.081   \\ [2pt]
  size 0.2 && 0.141 & 0.143 & 0.217 & 0.176   \\ [2pt]
  \hline \\
 \end{tabular}
 \end{center}
 \begin{scriptsize}
  1000 Monte Carlo  replications. ``Unc H'' is the average of the uncentered
  Hausman test statistic and `` JK H ''  is the average of the centered
  Hausman test statistic across Monte Carlo replications.
  ``size'' denotes the
  rejection rate for a nominal size of 5,10 and 20\%. 
 199 bootstrap  replications used for the jackknife correction.   
  \end{scriptsize}  
\end{table}

\begin{table}
  \caption{Power analysis, Probit model: DGP-IFE-P, OW-FE vs TW-GFE}
  \label{tab:sim_res_bin_FE_H1}
  \begin{center}
      \begin{tabular}{llcccc}
        \hline
        \\
 &&   \multicolumn{4}{c}{$\gamma = 1$}\\[2pt] 
&&  \multicolumn{2}{c}{T=10} &  \multicolumn{2}{c}{T=20}\\[2pt] 
$N$& & \multicolumn{2}{c}{$50$\quad$100$} &
                                            \multicolumn{2}{c}{$50$\quad$100$} \\[2pt] 
\hline \\
\multirow{4}*{$\sigma^2_{\zeta_t}=0.5$}  \\ [ 2pt] 
&Unc-H & 5.456 & 12.448 & 16.478 & 37.941 \\ 
 & Rejection rate  & 0.367 & 0.955 & 0.975 & 0.999 \\ 
 & JK-H & 3.491 & 4.661 & 5.266 & 7.343 \\ 
 & Rejection rate & 0.168 & 0.239 & 0.259 & 0.373 \\
\multirow{4}*{$\sigma^2_{\zeta_t}=1$}  \\ [ 2pt] 
&Unc-H & 5.596 & 12.733 & 18.353 & 41.801 \\ [ 2pt] 
&Rejection rate & 0.368 & 0.970 & 0.993 & 1.000 \\ [ 2pt] 
 & JK-H & 3.806 & 5.122 & 5.088 & 7.414 \\ [ 2pt] 
& Rejection rate & 0.174 & 0.258 & 0.288 & 0.434 \\  
  \multirow{4}*{$\sigma^2_{\zeta_t}=2$}  \\ [ 2pt] 
&Unc-H & 5.012 & 11.373 & 17.627 & 39.671 \\ [ 2pt] 
 & Rejection rate  & 0.291 & 0.915 & 0.978 & 1.000 \\ [ 2pt] 
 & JK-H & 3.498 & 4.754 & 5.309 & 7.643 \\ [ 2pt] 
 & Rejection rate & 0.160 & 0.247 & 0.313 & 0.483 \\  
  \hline \\
 \end{tabular}
 \end{center}
 \begin{scriptsize}
 1000 Monte Carlo (MC) replications. ``Unc H'' is the average of the uncentered
  Hausman test statistic, across MC replications. `` JK H ''  is the average of the centered
  Hausman test statistic, across MC replications.
  ``Rejection rate'' denotes the
  rejection rate for a nominal size of 5\%. ``$\sigma^2_{\zeta_t}$'' refers to the variance of the 
  time factor.
199 replications used for the jackknife correction.  
  \end{scriptsize}  
\end{table}

\subsection{Comparison with CRT and BBP tests}

Table \ref{CRT_test} reports the simulation results for the CRT
and BBP tests under the null and alternative hypotheses described by
DGP-FE-L and DGP-IFE-L. In particular, the CRT test can be implemented by
considering the null hypothesis of no factor structure with
time-constant individual effects and the alternative hypothesis of a
factor model with one latent factor (DGP-IFE-L).The test fails to attain
the correct size in short panels, as also reported by
\cite{castagnetti2015}.  The BBP test has good size properties in both
linear and nonlinear frameworks, while it displays remarkably low power, due
to the absence of a dynamic factor structure in DGP-IFE-L, as also discussed in
Subsection \ref{ALTRITEST}.

\begin{table}
\caption{Size and power analyses: Linear and probit model, CRT and BBP tests}
\label{CRT_test}
\begin{center}
  \begin{footnotesize}
\begin{tabular}{lcccccccccc}
\hline \\
\smallskip
&& \multicolumn{4}{c}{\textbf{DGP-FE}} &&\multicolumn{4}{c}{\textbf{DGP-IFE}} \\
\cline{3-6} \cline{8-11} \\
&&\multicolumn{2}{c}{T=10} & \multicolumn{2}{c}{T=20} &&\multicolumn{2}{c}{T=10} & \multicolumn{2}{c}{T=20}\\
         & $N$& 50    &100& 50   &100&& 50    &100& 50     &100 \\[3pt]
\cline{2-11}\\

Linear        &       &         &         &        &          &&$\sigma^2_{\zeta_t}=1$   &          &         &\\[5 pt]
&CRT&0.994 &0.994 &1.000& 1.000 && 0.996&   1.000&1.000& 1.000\\[3pt]
         &BBP&0.069 & 0.059& 0.06 & 0.063&& 0.105& 0.137 & 0.121 & 0.055  \\
        &       &         &         &        &          &&$\sigma^2_{\zeta_t}=3$  &          &         &\\[5 pt]
         &CRT&0.994 &0.994 &1.000& 1.000 && 1.000&   1.000&1.000& 1.000\\[3pt]
         &BBP&0.069 & 0.059& 0.06 & 0.063&& 0.056 & 0.067 &0.059  &0.068   \\
                 &       &         &         &        &          &&$\sigma^2_{\zeta_t}=5$   &          &         &\\[ 3 pt]
&CRT&0.994 &0.994 &1.000& 1.000 && 1.000&   1.000&1.000& 1.000\\[3pt]
         &BBP&0.069 & 0.059& 0.06 & 0.063&& 0.044& 0.03 & 0.055 & 0.066  \\
\cline{2-11}\\
Probit       &       &         &         &        &          &&$\sigma^2_{\zeta_t}=0.5$   &          &         &\\[5 pt]
&BBP&0.076& 0.073& 0.077& 0.061&&0.091 &0.1 & 0.07& 0.097 \\[5pt]  
             &       &         &         &        &          &&$\sigma^2_{\zeta_t}=2$   &          &         &\\[5 pt]       
       &       &         &         &        &          &&0.105&0.103&0.097&0.107 \\ 
      &       &         &         &        &          &&$\sigma^2_{\zeta_t}=2$   &          &         &\\[5 pt]       
 &&& & & &&0.122&0.11& 0.112& 0.11 \\[2pt]  
  \hline \\
\end{tabular}
\end{footnotesize}
\end{center}
\end{table}

\subsection{Testing for parallel trends in DID  }
In this Section we evaluate the performance of the proposed test when the null hypothesis is  the presence of a common trend
between treated and untreated units. 
We generate data from the linear model: 
\begin{eqnarray*}
 y_{it}& = & x_{it1}\th_1 + x_{it2}\th_2 + \xi_{i}  + \zeta_t + 
             \varepsilon_{it}, \\
  x_{it} & =  & \xi_{i}    + \zeta_t + N(0,4). \nonumber
\end{eqnarray*}
where  $\xi_{i} \sim N(1,1)$, the coefficients $\th = (\th_1,   \th_2)'$ are both equal to 1, $\zeta_t  \sim N(1,3)$ and $\varepsilon_{it}$ is an idiosyncratic standard normal error term.
We refer to this DGP as DGP-DID-0. 
After the de-meaning procedure, DGP-DID-0 traces back to a model where unobservable traits are time-varying.

 Under the alternative hypothesis we generate data according to the following model:
 \begin{eqnarray}
 y_{it}& = & x_{it1}\th_1 + x_{it2}\th_2 + \xi_{i}  +  \zeta_t [(1 - D_i)\delta_1 + D_i\delta_2] + 
             \varepsilon_{it}, \\
  x_{its} & =  & \xi_{i}    +   \zeta_t [(1 - D_i)\delta_1 + D_i\delta_2]+ N(0,4) \qquad s=1,2. \nonumber
\end{eqnarray}
where $D_i=1$ for a proportion $p=0.5$  of treated individuals and 0 otherwise, while $\delta_1=1$ and $\delta_2=2,3$.  In this way, we generate two different trends for treated and untreated units, ending up in a simplified version of a model with  heterogeneous trends: in this light, it can be seen as a factor model with two individual-specific loadings, $\delta_1$ and $\delta_2$. We refer to this DGP as DGP-DID-1.
  For each experiment, we  consider $N = 10,20$, $T = 50$, and $199$ bootstrap draws  in $1000$ Monte Carlo replications.
 
 Table \ref{tab:sim_res_lin_DID_H0_H1} reports results of the simulation study: the left part of the table refers to the size analysis, while the right part refers to the power analysis. We report results for the two versions of the test statistics, the centered (“JK H”) and the uncentered (“Unc. H”).
 The  centered version of the Hausman test exhibits a uniform p-value between experiments. Using second moments in the clustering procedure may introduce noise, resulting in the identification of a greater number of groups than there actually are.  We find that the choice of $\ga=1$  is optimal, as it helps alleviating the noise, keeping under control the number of groups.\footnote{We conjecture that the noise tends to disappear as the $N,T$ dimensions grow larger. In unreported simulations available upon request,  we find that smaller values of $\ga$ lead to a slightly non-uniform p-value.}. We highlight that, in this setting, the  test requires a large $T$ dimension in order to achieve proper size control. Conclusions drawn on the uncentered test statistics are, instead, not reliable.
  
Finally, the test exhibits good power properties,  increasing in the difference between the two factor loadings $\delta_1$ and $\delta_2$.
\begin{table}
  \caption{DID}
  \label{tab:sim_res_lin_DID_H0_H1}
  \begin{center}
      \begin{tabular}{llccllccc}
        \hline\\
        \emph{Size analysis} &&&  && \emph{Power analysis} &$\delta_1=1,\delta_2=2$&&                               \\ [5 pt]
 & &   \multicolumn{2}{c}{$\gamma = 1$}  &&&& \multicolumn{2}{c}{$\gamma = 1$} \\[2pt]
 $T=50$& & $N=10$ & $N=20$  &&   $T=50$& & $N=10$ & $N=20$   \\
\hline \\
Unc. H && 23.722 & 26.649  && Unc. H && 132.137 & 343.689   \\ 
  size.0.05 && 0.944 & 0.970 && Rejection rate && 0.999 & 1.000    \\ 
    size 0.1 && 0.960 & 0.982 & &  JK H && 16.722 & 14.076     \\ 
      size 0.2 && 0.983 & 0.994 & & Rejection rate && 0.808 & 0.753    \\ 
      \hline \\
      && &                          &&            &$\delta_1=1,\delta_2=3$  & &  \\[ 3pt]
  JK H && 2.227 & 1.976   &&       Unc. H && 147.852 & 391.046   \\
  size 0.05 && 0.058 & 0.043 &&    Rejection rate && 1.000 & 1.000   \\ 
  size 0.1 && 0.108 & 0.097 &&   JK H && 28.117 & 29.261     \\ 
  size 0.2 && 0.232 & 0.189 &&    Rejection rate && 0.880 & 0.863   \\ 
   \hline
 \end{tabular}
 \end{center}
 \begin{scriptsize}
  1000 Monte Carlo (MC) replications. ``Unc. H'' is the average of the uncentered
  Hausman test statistic, across MC replications. `` JK H ''  is the average of the centered
  Hausman test statistic, across MC replications. “size” refer to the empirical size of the test, for three  values of statistical significance (0.05,0.1,0.2).
  ``Rejection rate'' denotes the
  rejection rate for a nominal size of 5\%. $\delta_1$ and $\delta_2$ refer to the loadings of the
  two time trends.
   199 replications used for the jackknife correction.  
  \end{scriptsize}    
\end{table}
\section{Empirical applications}\label{empirical}
We evaluate the proposed test on real data in  two applications, one concerning determinants of housing prices, the other concerning intertemporal  decisions on labor
market participation of working women.

\subsection{Determinants of housing prices}
In this section we evaluate the proposed test  by revisiting the empirical application on housing prices in \cite{holly2010spatio}.
This study argues that  accounting for a  factor structure for unobservable traits gives a better understanding of determinants of housing prices in the U.S. These conclusions are confirmed by an extended study which examines  housing data at a metropolitan level \citep{baltagi2014further}. Similar findings are reported in the empirical section of  \cite{freemanweidner2023}, where the Authors use the \cite{giglio2016no}’s dataset on housing market. The presence of critical events, such as the financial crisis in 2008 may introduce structural breaks, suggesting the adoption of a factor structure in the specification. 
The outcome variable in exam is the growth rate of the  housing price index for $N=49$ States, observed for $T=35$ time occasions, from 1977 to 2011. We observe additional covariates such as the growth rate of population, the US State real cost of borrowing, net of real house price appreciation/depreciation and the inflation rate.

Table \ref{tab:app_lin_houses} reports results for the estimation of a linear model. We report estimated coefficients for FE, jackknife, TW-GFE with $\ga=1$ and its jackknife version. Moreover, we report the estimated coefficients for the IFE estimator \citep{Bai2009} with $r=1$ factor.
Let us start from the second panel of Table \ref{tab:app_lin_houses}, reporting the outcome of the test: the proposed test rejects the null hypothesis of time-invariant heterogeneity. The first panel in Table \ref{tab:app_lin_houses} shows the different outcomes of estimation: real cost of borrowing and inflation are deemed  statistically significant by all estimation methods, while the population is  recognized as significant by none of them, but the TW-GFE jackknife. The real cost of borrowing is associated to a decrease of the growth rate of housing price index, while the inflation has  a strong positive effect.

\begin{table} 
  \caption{Estimation results and test: house prices}
  \label{tab:app_lin_houses}
  \begin{center}

\begin{tabular}{lrrrrr}
\hline \\
 & ML & J & IFE & GFE  & JGFE \\ 
  \hline\\
population growth & 0.052 & 0.054 & -0.003 & 0.043 & 0.454 \\ [2 pt]
                             & (0.082) & (0.082) & (0.057) & (0.033) & (0.033) \\ [2 pt]
  real cost of borrowing & -0.781 & -0.776 & -0.941 & -0.881 & -0.736 \\  [2 pt]
                                  & (0.010) & (0.010) & (0.010) & (0.009) & (0.009) \\  [2 pt]
 inflation & 0.482 & 0.498 & 0.244 & 0.290 & 1.360 \\ [2 pt]
              & (0.020) & (0.020) & (0.061) & (0.016) & (0.016) \\  [2 pt]
  \hline \\ 
  \emph{Hausman Test}   & &&&& \\  [2 pt]
  JK H 									&  & &  & 12.507 &  \\ 
  Unc H &  &  && 373.609 & \\ 
 $\chi^2_3$&  &  &  & 7.815 &  \\ 
  \hline \\
  K & - & - & - & 2 & 2\\ 
  L & - & - & - & 15 & 15 \\ 
\hline \\
  \end{tabular}
   \end{center}
   \begin{scriptsize}
      Standard errors in parentheses. “IFE” is the \cite{Bai2009}’s estimator. ``Unc H'' is  the uncentered
  Hausman test statistic. `` JK H ''  is  the centered
  Hausman test statistic. ``$\chi^2_3$ crit.'' is the $95^{th}$
      percentile of the standard chi-squared distribution with 3 d.o.f. $K$ and $L$ are the number of groups for individuals and time occasions found by TW-GFE in the first step. 299 bootstrap replications used for the jackknife correction.  500 bootstrap replications used for standard errors of IFE estimator. $N=49$,$T=35$.
   \end{scriptsize}
\end{table}
\subsection{A model for determinants of inter-temporal occupational decisions of working women}

We apply the proposed Hausman test on a popular dataset concerning  inter-temporal labor supply
decisions of women, employed in many studies \citep[see, for instance,][]{fernandez2009fixed,dj2015}. Data are relative to
the employment status of $N=1461$ married women aged between 18 and 60
years in 1985, whose husbands were always employed in the period 1980-1988 (T=9) (PSID
waves 15-22).
Studies on relationship between fertility and employment make use of standard nonlinear models with specifications
that include one or more lags of the occupational status \citep{hyslop1999state}. We estimate a 
static logit model and compute the proposed test comparing the FE estimator with fixed effects at individual level and the TW-GFE.  Due to the assumptions on TW-GFE we can not estimate a dynamic model, while we  include all the other variables described in the empirical application in \cite{dj2015}: the number of kids of different ages and the logarithm of the yearly income of the husband.  
Table \ref{tab:app_nonlin_lfp} reports results of the  estimation, together with the outcome of the Hausman test. The first two columns refers to the one-way fixed effects.  We present the results of the  TW-GFE estimation using  three values of $\ga=0.2,0.5,1$, corresponding to a decreasing number of groups found in the data by the \emph{k-means} algorithm. 
The proposed test  fails to reject the null hypothesis of time-invariant heterogeneity, meaning that a logit model with individual fixed effects  gives right quantification of the parameters of interest. This suggested specification for the  unobserved heterogeneity is in line with that mostly used in the labour economics literature. 

As for the estimated coefficients, they are all in line with economic intuition: the presence of  small children and a high income of the husband are negatively associated with the labor force participation. Moreover, the coefficients associated to children in age 0-2, children in age 3-5 and  the income of the husband are statistically significant for all estimators (but for the TW-GFE with $\ga=1$ for the income).

\begin{table} 
  \caption{Estimation results and test: labour market participation }
  \label{tab:app_nonlin_lfp}
  \begin{center}
  \begin{footnotesize}
\begin{tabular}{lrrrrrrrr}
  \hline \\
 & ML & J & GFE 0.2 & GFE 0.5 & GFE 1 & JGFE 0.2 & JGFE 0.5 & JGFE 1 \\ 
  \hline \\
child (0-2) & -0.71 & -0.61 & -0.30 & -0.33 & -0.36 & -0.25 & -0.38 & -0.52 \\ 
          & (0.055) & (0.055) & (0.050) & (0.041) & (0.036) & (0.050) & (0.041) & (0.036) \\[2 pt]
  child (3-5) & -0.34 & -0.31 & -0.14 & -0.13 & -0.18 & -0.27 & -0.07 & -0.25 \\ 
          & (0.049) & (0.049) & (0.054) & (0.042) & (0.035) & (0.054) & (0.042) & (0.035) \\  [2 pt]
  child (6-17) & 0.01 & 0.01 & 0 & 0 & -0.03 & 0.02 & -0.04 & 0.07 \\ 
            & (0.035) & (0.035) & (0.037) & (0.029) & (0.023) & (0.037) & (0.029) & (0.023) \\ [2 pt] 
  income husband & -0.21 & -0.19 & -0.14 & -0.11 & -0.11 & -0.14 & 0.09 & -0.02 \\   
			 & (0.054) & (0.054) & (0.039) & (0.032) & (0.029) & (0.039) & (0.032) & (0.029) \\ [2 pt]
  \hline \\
 \emph{Hausman test} & &  &  &  &  &  & & \\ [2 pt]
 JK H & &  & 1.24 & 2.54 & 0.76 &  & & \\ 
Unc H &  & & 47.88 & 77.97 & 79.48 &  &  & \\ 
$\chi^2_4$&  &  & 9.49 & 9.49 & 9.49 &  &  &  \\ 
  \hline \\
  K & - & -& 268 & 82 & 31 & 268& 82 & 31 \\ 
  L & -& - & 7 & 5 & 4 & 7 & 5 & 4 \\ 
   \hline\\
   \end{tabular}
   \end{footnotesize}
   \end{center}
   \begin{scriptsize}
      Standard errors in parentheses. ``Unc H'' is  the uncentered
  Hausman test statistic. `` JK H ''  is  the centered
  Hausman test statistic. ``$\chi^2_4$ crit.'' is the $95^{th}$
      percentile of the standard chi-squared distribution with 4 d.o.f. $K$ and $L$ are the number of groups for individuals and time occasions found by TW-GFE in the first step.The dependent variable is not used in the clustering procedure.  299 bootstrap replications used for the jackknife correction.   $N=1461, T=9$
   \end{scriptsize}
\end{table}

\section{Final remarks}
\label{conclusion}

We propose a specification test for the form of the unobserved
heterogeneity in panel data models. The test is based on the
recently proposed TW-GFE approach and serves to
detect departures from the commonly assumed time-invariant specification.

The main advantage of our proposal is that it allows practitioners to
avoid the specification and estimation of models with complex forms of
time-varying heterogeneity, which might pose identification and
computational problems in both linear and nonlinear models. By
contrast, the TW-GFE approach is a rather simple non-iterative
two-step strategy, involving unsupervised clustering in the first step
and estimation of group effects in the second.
The proposed approach is a Hausman test constrasting the ML and TW-GFE estimators.
We show that, under the null hypothesis of time-invariant effects, the TW-GFE
estimator is equivalent  to the OW-GFE, implying that the Hausman
statistic has asymptotic $\chi^2$ distribution.
This distribution is however  non-centered because of the bias arising
from incidental parameters, for both the ML (at least in the non-linear case) and TW-GFE estimators, and from the
approximation error induced by the discretization of the unobserved
heterogeneity, that arises with the GFE
approach. 
We
make the statistic pivotal by reducing the bias of the vector
of contrasts with leave-one-out jackknife.
 Parametric bootstrap is then used  to estimate the variance of the
vector of contrasts.

The two empirical applications considered suggest that the test is reliable with real world data 
and suggest specifications for unobserved heterogeneity in line with the relevant literature.
Finally, the proposed test also emerges as a viable alternative to existing
procedures with short panel datasets and can also be applied to test for the common trend assumption in the DiD setting.

\section*{Acknowledgments}
 Claudia Pigini and Alessandro 
   Pionati would like to acknowledge the financial support by  Project Title: ``The use of the Grouped Fixed Effects estimator in panel data analysis addressing unobserved heterogeneity”, financial coverage D.D. MUR 47/2025, CUP I33C25000280001.

\clearpage
\bibliography{biblio3.bib}
\bibliographystyle{apalike}
\clearpage

\appendix

\section{Proofs}
\subsection{Proof of Theorem   \ref{th:tw-gfe}}\label{p_th_tw_gfe}
\begin{lemma}
  \label{lemma:L1}
  Under Assumptions \ref{ass:uh}-\ref{ass:mom} and  when data are sampled
  under $H_0$, that is $\psi(\lambda_{t0})
  = \psi(\lambda_0)$, $\hat{L} = 1$ w.p.a. 1. as $N\to\infty$
  \begin{proof} Let $L = 1$ and define
    $\hat{w}_C = \frac{1}{T}\sum_{t=1}^T w_t$, that is the only
    centroid.  The {\em kmeans} objective function is then
    $\hat{Q}(1) = \frac{1}{T}\sum_{t=1}^T|| w_t - \hat{w}_C ||^2$ and
    by the Cauchy-Schwartz inequality we have
    \[
\hat{Q}(1) = \frac{1}{T}\sum_{t=1}^T|| w_t -
    \hat{w}_C ||^2 = \frac{1}{T}\sum_{t=1}^T|| w_t - \psi(\lambda_0) +
    \psi(\lambda_0) -  \hat{w}_C ||^2
  \]
  \[
\leq \frac{1}{T}\sum_{t=1}^T|| w_t - \psi(\lambda_0)||^2 +
\frac{1}{T}\sum_{t=1}^T|| \psi(\lambda_0) - \hat{w}_C||^2 +
\frac{2}{T}|| \psi(\lambda_0) - \hat{w}_C|| \sum_{t=1}^T|| w_t - \psi(\lambda_0)||
  \]  
  As $N \rightarrow \infty$, the first term on the rhs is
  $O_p\left(\frac{1}{N} \right)$ by Assumption \ref{ass:mom}. It is
  also an infeasible estimator of $V_w$, so that it can be written as
  $\mathrm{E}\left[||w_t - \psi(\lambda_0) ||^2\right] +
  o_p\left(\frac{1}{N}\right)$. The second and third terms involving
    $|| \psi(\lambda_0) - \hat{w}_C||$ are $o_p(1)$, since
    $\mathrm{plim}_{N\to\infty} w_t = \psi(\lambda_0)$ again by
    Assumption \ref{ass:mom} under $H_0$ so that
    $\mathrm{plim}_{N\to\infty} \hat{w}_C = \frac{1}{T} T
    \psi(\lambda_0) = \psi(\lambda_0)$. Intuitively, there is no
    approximation error from the {\em kmeans}, as the heterogeneity is
    discrete with a unique support point.

Since $\hat{V}_w = \mathrm{E}\left[||w_t - \psi(\lambda_0) ||^2
\right] + o_p\left(\frac{1}{N} \right)$, then $\hat{Q}(1)$ is equal or
smaller than
$\hat{V}_w$ as $N\to\infty$. As the {\em kmeans} objective function
in nondecreasing in $L$, the deterministic rule 
in Remark \ref{remark:kandl} will deliver $\hat{L} = 1$ w.p.a. 1 as
$N\to\infty$. 
  \end{proof}
\end{lemma}

\begin{lemma}
  \label{lemma:L2}
  Under Assumptions \ref{ass:uh}-\ref{ass:reg},\ref{ass:mom} when data are sampled
  under $H_0$ and when the number of groups are determined according
  to the rule in Remark \ref{remark:kandl} (i), we have that
  \begin{equation}
    \label{eq:lemma2}
    P(||\tilde{\th}-\dot{\th}||<\varepsilon) \to 1  \quad \text{as} \quad N\to\infty,  
  \end{equation}
  where $\dot{\th}$ is the OW-GFE and $\tilde{\th}$ is the TW-GFE.

  \begin{proof}
    With reference to the TW-GFE, the two {\em kmeans} procedures on individual- and
    time averages are performed independently, so that we have
    \begin{equation}
      \tilde{\th} \equiv \dot{\th} \iff \hat{L} = 1.
    \end{equation}
    Define the event
    \begin{equation}
      D_N = \mathrm{1}\{||\tilde{\th}-\dot{\th}|| = 0 \} = \begin{cases}1 \quad \text{if} \quad \hat{L} = 1 \\ 0 \quad \text{if} \quad \hat{L} > 1\end{cases}
    \end{equation}
    By Lemma \ref{lemma:L1}, $D_N \convp 1$ as $N\to\infty$ and the
    result follows.
  \end{proof}
\end{lemma}

 \begin{proof}[Proof of Theorem \ref{th:tw-gfe}]
    By Lemma \ref{lemma:L2} we have that
    $P(||\tilde{\th}-\dot{\th}||<\varepsilon) \to 1$ as
    $N,T \to \infty$. Combining this result with Theorem 1 and
    Corollary 2 in \cite{blm2022}, the proof follows as a direct
    application of Theorem 2.7 (iv) in \cite{van2000asymptotic}. 
    \end{proof}
 \subsection{Proof of Theorem \ref{th:jack_delta}}\label{th_2_j}
Proof of Theorem \ref{th:jack_delta} directly follows from  Theorem 3 in \cite{HN2004}, which concerns  the asymptotic distribution of the jackknife estimator in FE models. In order to prove our result we need to verify Conditions 1 to 4 in \cite{HN2004}, HN1-HN4 henceforth, needed for their corresponding results. These requirements are standard regularity conditions for a well-posed likelihood maximization problem.
\medskip

\noindent
Condition HN1 is equivalent to Assumption \ref{ass:rect}, and it is the rectangular array asymptotics requirement. 
Parts (i) and (ii) of Condition HN2 are equivalent to our Assumption \ref{ass:reg}-(i); part (iii) of HN2 is  directly stated in Assumption \ref{ass:reg}-(iv).
 Condition HN3 is an identification assumption and is equivalent to Assumption \ref{ass:reg}-(ii), as it states that log-likelihood function has a unique maximum. 
Condition HN4-(i) is   directly stated in  Assumption \ref{ass:reg}-(iv), Condition HN4, part (ii) and (iii) are properties of the Hessian and equivalent those  described by Assumption \ref{ass:reg}-(ii).
As all the Conditions required for the result of Theorem 3 in  \cite{HN2004} are verified, the result of 
Theorem \ref{th:jack_delta} follows.

\section{Additional simulation results}\label{tabelle_solo_stimatori}
 \begin{landscape} 
\begin{table}
  \caption{Size analysis: DGP-FE-L, simulation statistics of the estimators}
  \label{tab:sim_res_lin_FE_stim}
  \begin{center}
    \begin{footnotesize}
      \begin{tabular}{llcccccccccccc}
        \hline
        \\
& &   \multicolumn{4}{c}{$\gamma = 0.2$} & \multicolumn{4}{c}{$\gamma = 0.5$} & \multicolumn{4}{c}{$\gamma = 1$} \\[2pt]
& &  \multicolumn{2}{c}{T=10} &  \multicolumn{2}{c}{T=20} &  \multicolumn{2}{c}{T=10}&  \multicolumn{2}{c}{T=20}&  \multicolumn{2}{c}{T=10}&  \multicolumn{2}{c}{T=20}\\[2pt]
&$N$ & \multicolumn{2}{c}{$50$\quad$100$} & \multicolumn{2}{c}{$50$\quad$100$} & \multicolumn{2}{c}{$50$\quad$100$} & \multicolumn{2}{c}{$50$\quad$100$} & \multicolumn{2}{c}{$50$\quad$100$} & \multicolumn{2}{c}{$50$\quad$100$}\\[2pt]
\hline \\  
\multirow{2}*{ML}  \\ [ 2pt] 
&Bias & 0.000 & -0.000 & -0.000 & -0.000 & 0.000 & -0.000 & -0.000 & -0.000 & 0.000 & -0.000 & -0.000 & 0.000 \\ [ 2pt] 
 & SD & 0.047 & 0.034 & 0.033 & 0.023 & 0.047 & 0.034 & 0.033 & 0.023 & 0.047 & 0.034 & 0.033 & 0.024 \\ [ 2pt] 
  \multirow{2}*{ML2}  \\ [ 2pt] 
  &Bias & -0.001 & -0.000 & -0.000 & 0.000 & -0.001 & -0.000 & -0.000 & 0.000 & -0.001 & -0.000 & -0.000 & -0.000 \\ [ 2pt] 
 & SD & 0.045 & 0.032 & 0.032 & 0.023 & 0.045 & 0.032 & 0.032 & 0.023 & 0.045 & 0.032 & 0.032 & 0.023 \\ [ 2pt] 
  \hline \\
  \multirow{2}*{J}  \\ [ 2pt] 
  &Bias & 0.000 & -0.000 & -0.000 & -0.000 & 0.000 & -0.000 & -0.000 & -0.000 & 0.000 & -0.000 & -0.000 & 0.000 \\ [ 2pt] 
  &SD & 0.047 & 0.034 & 0.033 & 0.023 & 0.047 & 0.034 & 0.033 & 0.023 & 0.047 & 0.034 & 0.033 & 0.024 \\ [ 2pt] 
    \multirow{2}*{J2}  \\ [ 2pt] 
  &Bias & -0.001 & -0.000 & -0.001 & 0.000 & -0.001 & -0.000 & -0.001 & 0.000 & -0.001 & -0.000 & -0.001 & -0.000 \\ [ 2pt] 
 & SD & 0.045 & 0.032 & 0.032 & 0.023 & 0.045 & 0.032 & 0.032 & 0.023 & 0.045 & 0.032 & 0.032 & 0.023 \\ [ 2pt] 
  \hline \\
    \multirow{2}*{GFE}  \\ [ 2pt] 
  &Bias & -0.015 & -0.016 & -0.008 & -0.009 & -0.038 & -0.040 & -0.021 & -0.023 & -0.075 & -0.080 & -0.043 & -0.045 \\ [ 2pt] 
  &SD & 0.047 & 0.035 & 0.034 & 0.023 & 0.051 & 0.037 & 0.035 & 0.025 & 0.061 & 0.043 & 0.041 & 0.029 \\ [ 2pt] 
    \multirow{2}*{GFE2}  \\ [ 2pt] 
  &Bias & -0.015 & -0.016 & -0.008 & -0.008 & -0.039 & -0.041 & -0.021 & -0.022 & -0.077 & -0.080 & -0.044 & -0.046 \\ [ 2pt] 
  &SD & 0.046 & 0.032 & 0.033 & 0.023 & 0.050 & 0.036 & 0.034 & 0.024 & 0.060 & 0.042 & 0.038 & 0.028 \\[ 2pt] 
   \hline \\
    \multirow{2}*{GFE J}  \\ [ 2pt] 
  &Bias & -0.001 & -0.002 & 0.008 & -0.002 & -0.002 & -0.006 & 0.011 & -0.014 & -0.015 & -0.018 & 0.001 & -0.006 \\ [ 2pt] 
  &SD & 0.126 & 0.097 & 0.132 & 0.104 & 0.224 & 0.166 & 0.262 & 0.194 & 0.369 & 0.268 & 0.442 & 0.317 \\ [ 2pt] 
    \multirow{4}*{GFE J2}  \\ [ 2pt] 
  &Bias & 0.002 & -0.003 & 0.003 & 0.003 & -0.007 & -0.014 & 0.010 & -0.003 & -0.029 & -0.028 & -0.011 & -0.010 \\ [ 2pt] 
  &SD & 0.120 & 0.093 & 0.142 & 0.100 & 0.232 & 0.170 & 0.257 & 0.193 & 0.367 & 0.270 & 0.435 & 0.316 \\ [ 2pt] 
  &K & 38.723 & 72.467 & 41.156 & 78.419 & 31.317 & 54.613 & 35.226 & 63.344 & 24.137 & 38.580 & 28.945 & 48.436 \\ [ 2pt] 
  &L & 1.000 & 1.000 & 1.000 & 1.000 & 1.000 & 1.000 & 1.000 & 1.000 & 1.000 & 1.000 & 1.000 & 1.000 \\ 
  \hline \\
  \end{tabular}
  \end{footnotesize}
  \end{center}
\begin{scriptsize}
  1000 Monte Carlo (MC) replications. ``Bias ``is the mean bias, ``SD'' is the standard deviation of the estimator. ``K'' and ``L'' average number of groups. Estimators: ``ML'', ``J'' is the leave-one-out jackknife, ``GFE'' is the TW-GFE, ``GFE J'' is the leave-one-out jackknife applied on TW-GFE.   
 \end{scriptsize}  
\end{table}
\end{landscape}

\begin{landscape}
\begin{table}
  \caption{Power analysis: DGP-IFE-L, simulation statistics of the estimators}
  \label{tab:sim_res_lin_FE_H1_stim}
  \begin{center}
            \begin{tabular}{llcccccccccccc}
        \hline        \\
& &   \multicolumn{4}{c}{$\sigma=1$} & \multicolumn{4}{c}{$\sigma=3$} & \multicolumn{4}{c}{$\sigma=5$} \\[2pt]
& &  \multicolumn{2}{c}{T=10} &  \multicolumn{2}{c}{T=20} &  \multicolumn{2}{c}{T=10}&  \multicolumn{2}{c}{T=20}&  \multicolumn{2}{c}{T=10}&  \multicolumn{2}{c}{T=20}\\[2pt]
&$N$ & \multicolumn{2}{c}{$50$\quad$100$} & \multicolumn{2}{c}{$50$\quad$100$} & \multicolumn{2}{c}{$50$\quad$100$} & \multicolumn{2}{c}{$50$\quad$100$} & \multicolumn{2}{c}{$50$\quad$100$} & \multicolumn{2}{c}{$50$\quad$100$}\\[2pt]
\hline \\ 
\multirow{2}*{ML}  \\ [ 2pt] 
&Bias & 0.257 & 0.255 & 0.255 & 0.258 & 0.263 & 0.273 & 0.266 & 0.274 & 0.264 & 0.272 & 0.278 & 0.270 \\ 
 & SD & 0.110 & 0.078 & 0.082 & 0.062 & 0.263 & 0.196 & 0.186 & 0.135 & 0.439 & 0.298 & 0.298 & 0.218 \\ 
 \multirow{2}*{ML2}  \\ [ 2pt] 
  &Bias & 0.252 & 0.257 & 0.255 & 0.255 & 0.273 & 0.267 & 0.271 & 0.267 & 0.275 & 0.264 & 0.262 & 0.267 \\ 
  &SD & 0.109 & 0.079 & 0.084 & 0.061 & 0.265 & 0.193 & 0.186 & 0.136 & 0.434 & 0.299 & 0.295 & 0.214 \\ 
  \multirow{2}*{J}  \\ [ 2pt] 
  &Bias & 0.261 & 0.259 & 0.257 & 0.260 & 0.264 & 0.274 & 0.267 & 0.275 & 0.264 & 0.272 & 0.278 & 0.270 \\ 
  &SD & 0.110 & 0.078 & 0.082 & 0.062 & 0.263 & 0.195 & 0.186 & 0.135 & 0.438 & 0.298 & 0.298 & 0.218 \\ 
  \multirow{2}*{J2}  \\ [ 2pt] 
  &Bias & 0.256 & 0.261 & 0.257 & 0.256 & 0.273 & 0.268 & 0.271 & 0.267 & 0.275 & 0.264 & 0.262 & 0.267 \\ 
  &SD & 0.109 & 0.079 & 0.084 & 0.061 & 0.265 & 0.193 & 0.185 & 0.136 & 0.433 & 0.298 & 0.295 & 0.214 \\ 
  \multirow{2}*{GFE}  \\ [ 2pt] 
  &Bias & -0.116& -0.190 & -0.041 & -0.116 & -0.065 & -0.123 & -0.098 & -0.182 & 0.004 & -0.040 & 0.001 & -0.053 \\ 
  &SD & 0.177 & 0.147 & 0.128 & 0.109 & 0.369 & 0.371 & 0.314 & 0.336 & 0.482 & 0.411 & 0.388 & 0.374 \\ 
  \multirow{2}*{GFE2}  \\ [ 2pt] 
  &Bias & -0.117 & -0.190 & -0.041 & -0.115 & -0.068 & -0.122 & -0.094 & -0.181 & 0.029 & -0.050 & -0.003 & -0.063 \\ 
  &SD & 0.178 & 0.151 & 0.129 & 0.108 & 0.377 & 0.362 & 0.326 & 0.336 & 0.474 & 0.427 & 0.382 & 0.376 \\ 
  \multirow{2}*{GFE J}  \\ [ 2pt] 
  &Bias & -0.056 & -0.122 & 0.073 & -0.028 & -0.099 & -0.165 & -0.137 & -0.303 & 0.057 & -0.049 & -0.221 & -0.236 \\ 
  &SD & 0.894 & 0.741 & 1.197 & 1.000 & 1.366 & 1.164 & 2.228 & 1.993 & 1.488 & 1.316 & 2.181 & 1.876 \\ 
  \multirow{4}*{GFE J2}  \\ [ 2pt] 
  &Bias & -0.025 & -0.127 & 0.053 & 0.030 & -0.134 & -0.180 & -0.160 & -0.289 & 0.094 & -0.039 & -0.139 & -0.221 \\ 
  &SD & 0.914 & 0.722 & 1.165 & 1.001 & 1.427 & 1.177 & 2.169 & 2.006 & 1.522 & 1.262 & 2.057 & 1.946 \\ 
  &K & 14.348 & 20.039 & 18.592 & 26.685 & 5.112 & 5.989 & 5.710 & 6.665 & 3.721 & 4.185 & 3.772 & 4.177 \\ 
  &L & 4.203 & 4.953 & 5.252 & 6.468 & 4.793 & 5.552 & 6.241 & 7.542 & 4.865 & 5.658 & 6.337 & 7.729 \\ 
 \hline \\
 \end{tabular}
 \end{center}
 \begin{scriptsize}
  1000 Monte Carlo (MC) replications. Statistics for different values of the standard deviation of time factor $\sigma_{\zeta_t}$. ``Bias ``is the mean bias, ``SD'' is the standard deviation of the estimator. ``K'' and ``L'' average number of groups. Estimators: ``ML'', ``J'' is the leave-one-out jackknife, ``GFE'' is the TW-GFE, ``GFE J''is the leave-one-out jackknife applied to TW-GFE.   
  \end{scriptsize}  
\end{table}
\end{landscape}

\begin{table}
  \caption{Size analysis: DGP-FE-P, simulation statistics of the estimators}
  \label{tab:sim_res_bin_FE_stim}
  \begin{center}
      \begin{tabular}{llcccc}
 \hline \\
& &   \multicolumn{4}{c}{$\gamma = 1$} \\[2pt]
& &  \multicolumn{2}{c}{T=10} &  \multicolumn{2}{c}{T=20} \\[2pt]
$N$& & \multicolumn{2}{c}{$50$\quad$100$} & \multicolumn{2}{c}{$50$\quad$100$} \\[2pt]
\hline \\  
\multirow{2}*{ML}  \\ [ 2pt] 
&Bias & 0.173 & 0.161 & 0.080 & 0.078 \\ 
&  SD & 0.239 & 0.165 & 0.146 & 0.111 \\ 
     \hline \\
  \multirow{2}*{ML2}  \\ [ 2pt] 
&Bias & 0.184 & 0.161 & 0.078 & 0.078 \\ 
&  SD  & 0.237 & 0.167 & 0.151 & 0.104 \\ 
     \hline \\
  \multirow{2}*{J}  \\ [ 2pt] 
&Bias & -0.035 & -0.028 & -0.007 & -0.004 \\ 
&  SD  & 0.192 & 0.137 & 0.134 & 0.102 \\ 
     \hline \\
    \multirow{2}*{J2}  \\ [2pt] 
&Bias& -0.028 & -0.028 & -0.009 & -0.004 \\ 
&  SD  & 0.190 & 0.137 & 0.137 & 0.096 \\ 
     \hline \\
    \multirow{2}*{GFE}  \\ [ 2pt] 
&Bias & -0.107 & -0.121 & -0.070 & -0.080 \\ 
&  SD & 0.183 & 0.127 & 0.129 & 0.095 \\ 
     \hline \\
    \multirow{2}*{GFE2}  \\ [ 2pt]  
&Bias& -0.101 & -0.121 & -0.072 & -0.081 \\ 
&  SD  & 0.183 & 0.126 & 0.132 & 0.092 \\ 
     \hline \\
    \multirow{2}*{GFE J}  \\ [ 2pt]   
&Bias & -0.047 & -0.077 & -0.010 & -0.037 \\ 
&  SD  & 0.549 & 0.401 & 0.632 & 0.455 \\ 
     \hline \\
    \multirow{4}*{GFE J2}  \\ [ 2pt]  
&Bias& -0.039 & -0.085 & -0.019 & -0.044 \\ 
&  SD  & 0.554 & 0.379 & 0.644 & 0.458 \\ 
&  K & 7.825 & 9.242 & 10.793 & 13.440 \\ 
&  L & 1.001 & 1.001 & 1.000 & 1.001 \\ 
    \hline \\
  \end{tabular}
  \end{center}
\begin{scriptsize}
  1000 Monte Carlo (MC) replications. ``Bias ``is the mean bias, ``SD'' is the standard deviation of the estimator. ``K'' and ``L'' average number of groups. Estimators: ``ML'', ``J'' is the leave-one-out jackknife, ``GFE'' is the TW-GFE, ``GFE J'' is the leave-one-out jackknife applied to TW-GFE.   
 \end{scriptsize}  
\end{table}

\begin{table}
  \caption{ Power analysis: DGP-IFE-P, simulation statistics of the estimators}
  \label{tab:sim_res_bin_FE_H1_stim}
  \begin{footnotesize}
  \begin{center}
            \begin{tabular}{llcccccccccccc}
        \hline        \\
& & \multicolumn{4}{c}{$\sigma_{\zeta_t}^2=0.5$} &  \multicolumn{4}{c}{$\sigma_{\zeta_t}^2=1$} & \multicolumn{4}{c}{$\sigma_{\zeta_t}^2=2$}  \\[2pt]
& & \multicolumn{2}{c}{T=10} &  \multicolumn{2}{c}{T=20} & \multicolumn{2}{c}{T=10} &  \multicolumn{2}{c}{T=20} &  \multicolumn{2}{c}{T=10}&  \multicolumn{2}{c}{T=20}\\[2pt]
&$N$ & \multicolumn{2}{c}{$50$\quad$100$} & \multicolumn{2}{c}{$50$\quad$100$} & \multicolumn{2}{c}{$50$\quad$100$} & \multicolumn{2}{c}{$50$\quad$100$} & \multicolumn{2}{c}{$50$\quad$100$} & \multicolumn{2}{c}{$50$\quad$100$}\\[2pt]
\hline \\ 
\multirow{2}*{ML}  \\ [ 2pt] 
&Bias & 0.502 & 0.489 & 0.368 & 0.364 & 0.580 & 0.565 & 0.432 & 0.428 & 0.647 & 0.632 & 0.487 & 0.482 \\ 
 & SD & 0.253 & 0.181 & 0.148 & 0.117 & 0.266 & 0.189 & 0.152 & 0.118 & 0.276 & 0.201 & 0.155 & 0.121 \\ 
 \multirow{2}*{ML2}  \\ [ 2pt] 
&Bias & 0.493 & 0.485 & 0.373 & 0.362 & 0.567 & 0.560 & 0.438 & 0.423 & 0.646 & 0.625 & 0.496 & 0.478 \\ 
 & SD & 0.245 & 0.184 & 0.151 & 0.116 & 0.251 & 0.189 & 0.151 & 0.116 & 0.268 & 0.196 & 0.157 & 0.116 \\ 
  \multirow{2}*{J}  \\ [ 2pt] 
 & Bias & 0.210 & 0.222 & 0.254 & 0.257 & 0.250 & 0.264 & 0.307 & 0.310 & 0.278 & 0.297 & 0.349 & 0.353 \\ 
  &SD & 0.190 & 0.135 & 0.131 & 0.103 & 0.189 & 0.136 & 0.133 & 0.103 & 0.194 & 0.143 & 0.133 & 0.104 \\
  \multirow{2}*{J2}  \\ [ 2pt] 
&  Bias& 0.203 & 0.219 & 0.259 & 0.255 & 0.239 & 0.260 & 0.312 & 0.306 & 0.277 & 0.291 & 0.357 & 0.349 \\ 
 & SD & 0.183 & 0.140 & 0.134 & 0.104 & 0.182 & 0.137 & 0.131 & 0.102 & 0.188 & 0.141 & 0.133 & 0.100 \\ 
   \multirow{2}*{GFE}  \\ [ 2pt] 
& Bias & 0.158 & 0.129 & 0.114 & 0.088 & 0.195 & 0.164 & 0.142 & 0.115 & 0.240 & 0.213 & 0.184 & 0.158 \\ 
 & SD & 0.188 & 0.129 & 0.129 & 0.096 & 0.198 & 0.136 & 0.135 & 0.098 & 0.202 & 0.147 & 0.141 & 0.109 \\ 
   \multirow{2}*{GFE2}  \\ [ 2pt] 
 & Bias & 0.150 & 0.126 & 0.118 & 0.086 & 0.182 & 0.160 & 0.146 & 0.112 & 0.236 & 0.210 & 0.191 & 0.157 \\ 
 & SD & 0.189 & 0.131 & 0.134 & 0.096 & 0.189 & 0.139 & 0.132 & 0.100 & 0.194 & 0.146 & 0.143 & 0.106 \\ 
  \multirow{2}*{GFE J}  \\ [ 2pt] 
&Bias& 0.097 & 0.075 & 0.069 & 0.046 & 0.116 & 0.100 & 0.102 & 0.065 & 0.166 & 0.137 & 0.117 & 0.091 \\ 
  &SD & 0.343 & 0.253 & 0.315 & 0.233 & 0.379 & 0.264 & 0.348 & 0.255 & 0.398 & 0.297 & 0.414 & 0.317 \\ 
 \multirow{4}*{GFE J2}  \\ [ 2pt] 
 & Bias & 0.080 & 0.080 & 0.066 & 0.053 & 0.098 & 0.086 & 0.108 & 0.071 & 0.156 & 0.146 & 0.116 & 0.087 \\ 
  &SD & 0.351 & 0.250 & 0.310 & 0.241 & 0.378 & 0.278 & 0.329 & 0.258 & 0.387 & 0.293 & 0.427 & 0.336 \\ 
  &  K & 3.382 & 3.587 & 4.252 & 4.566 & 2.827 & 2.946 & 3.525 & 3.717 & 2.296 & 2.367 & 2.835 & 2.946 \\ 
  &L & 3.061 & 3.649 & 3.625 & 4.506 & 3.448 & 4.097 & 4.156 & 5.173 & 3.808 & 4.462 & 4.677 & 5.760 \\ 
 \hline \\
 \end{tabular}
\end{center}
 \end{footnotesize}
 \begin{scriptsize}
  1000 Monte Carlo (MC) replications. Statistics for different values of the variance of time factor $\sigma_{\zeta_t}^2$. ``Bias ``is the mean bias, ``SD'' is the standard deviation of the estimator. ``K'' and ``L'' average number of groups. Estimators: ``ML'', ``J'' is the leave-one-out jackknife, ``GFE'' is the TW-GFE, ``GFE J'' is the leave-one-out jackknife applied to TW-GFE.   
  \end{scriptsize}  
\end{table}

\end{document}